\documentclass[11pt]{article}
\usepackage[margin=1in]{geometry}
\usepackage{amsmath,amsthm,amssymb,stmaryrd,enumitem,xspace,comment}
\usepackage{titlesec}
\usepackage[usenames, dvipsnames]{xcolor}
\usepackage{graphicx}
\usepackage{tikz}
\newif\ifpazo
\usepackage{mathpazo} \pazotrue
\DeclareMathAlphabet{\mathcal}{OMS}{cmsy}{m}{n}
\DeclareMathAlphabet{\mathbfcal}{OMS}{cmsy}{b}{n}
\usepackage{multirow}
\usepackage{algorithm}
\usepackage[noend]{algpseudocode}
\usepackage{authblk}
\usepackage{url}
\usepackage[colorlinks, linkcolor=BrickRed, citecolor=blue]{hyperref}
\usepackage{cleveref}

\crefname{appendix}{Section}{Sections}
\Crefname{lem}{Lemma}{Lemmas}

\newtheorem{note}{Observation}[section]

\def\clap#1{\hbox to 0pt{\hss#1\hss}}

\newcommand{\xqedhere}[2]{\rlap{\hbox to#1{\hfil\llap{\ensuremath{#2}}}}}

\newcommand{\eps}{\varepsilon}

\ifpazo  
\else  \fi

\newcommand{\Vnc}{V_{\textup{NC}}}
\newcommand{\lambdaM}{\lambda_{\textup{M}}}
\newcommand{\lambdaU}{\lambda_{\textup{U}}}
\newcommand{\advntg}{\alpha}
\renewcommand{\ge}{\geqslant}
\renewcommand{\le}{\leqslant}
\renewcommand{\geq}{\geqslant}
\renewcommand{\leq}{\leqslant}

\newtheorem{theorem}{Theorem}[section]

\newtheorem{lem}[theorem]{Lemma}

\theoremstyle{definition}

\theoremstyle{remark}

\DeclareMathOperator{\poly}{poly}
\DeclareMathOperator{\polylog}{polylog}

\title{Maximum Matching in Two, Three, and a Few More Passes Over Graph Streams}
\author[1]{Sagar Kale} \author[2]{Sumedh Tirodkar} \affil[1]{Dartmouth College.
  Email: sag (at) cs.dartmouth.edu} \affil[2]{TIFR, Mumbai.  Email:
  sumedh.tirodkar (at) tifr.res.in}
\begin{document}
\date{}
\maketitle

\begin{abstract} 
  \renewcommand{\baselinestretch}{1.05} We consider the maximum matching problem
  in the semi-streaming model formalized by Feigenbaum, Kannan, McGregor, Suri,
  and Zhang~\cite{fgnbm} that is inspired by giant graphs of today.  As our main
  result, we give a two-pass $(1/2 + 1/16)$-approximation algorithm for
  triangle-free graphs and a two-pass $(1/2 + 1/32)$-approximation algorithm for
  general graphs; these improve the approximation ratios of $1/2 + 1/52$ for
  bipartite graphs and $1/2 + 1/140$ for general graphs by Konrad, Magniez, and
  Mathieu~\cite{kmm}.  In three passes, we are able to achieve approximation
  ratios of $1/2 + 1/10$ for triangle-free graphs and $1/2 + 1/19.753$ for
  general graphs.  We also give a multi-pass algorithm where we bound the number
  of passes \emph{precisely}---we give a $(2/3 -\eps)$-approximation algorithm
  that uses $2/(3\eps)$ passes for triangle-free graphs and $4/(3\eps)$ passes
  for general graphs.  Our algorithms are simple and combinatorial, use
  $O(n \log n)$ space, and (can be implemented to) have $O(1)$ update time per
  edge.

  For general graphs, our multi-pass algorithm improves the best known
  \emph{deterministic} algorithms in terms of the number of passes:
  \begin{itemize}
    \item Ahn and Guha~\cite{ahnguha} give a $(2/3 - \eps)$-approximation
    algorithm that uses $O(\log(1/\eps)/\eps^2)$ passes, whereas our
    $(2/3 - \eps)$-approximation algorithm uses $4/(3\eps)$ passes;
    \item they also give a $(1-\eps)$-approximation algorithm that uses
    $O(\log n \cdot \poly(1/\eps))$ passes, where $n$ is the number of vertices
    of the input graph; although our algorithm is $(2/3 - \eps)$-approximation,
    our number of passes do not depend on $n$.
  \end{itemize}

  Earlier multi-pass algorithms either have a large constant inside big-$O$
  notation for the number of passes~\cite{Eggert2012} or the constant cannot be
  determined due to the involved analysis~\cite{mwms,ahnguha}, so our multi-pass
  algorithm should use much fewer passes for approximation ratios bounded
  slightly below $2/3$.
\end{abstract}

 \thispagestyle{empty}
 \addtocounter{page}{-1}
 \newpage

\section{Introduction}
\label{sec:intro}
Maximum matching is a well-studied problem in a variety of computational models.
We consider it in the semi-streaming model formalized by Feigenbaum, Kannan,
McGregor, Suri, and Zhang~\cite{fgnbm} that is inspired by generation of
ginormous graphs in recent times.  A graph stream is an (adversarial) sequence
of the edges of a graph, and a semi-streaming algorithm must access the edges in
the given order and use $O(n\polylog n)$ space only, where $n$ is the number of
vertices; note that a matching can have size $\Omega(n)$, so $\Omega(n\log n)$
space is necessary.  The number of times an algorithm goes over a stream of
edges is called the number of \emph{passes}.  A trivial $(1/2)$-approximation
algorithm that can be easily implemented as a one-pass semi-streaming algorithm
is to output a maximal matching.
Since the formalization of the semi-streaming model more than a decade ago, the
problem of finding a better than $(1/2)$-approximation algorithm or proving that
one cannot do better has baffled researchers~\cite{subl60}.  In a step towards
resolving this, Goel, Kapralov, and Khanna~\cite{gkk} proved that for any
$\eps >0$, a one-pass semi-streaming $(2/3 + \eps)$-approximation algorithm does
not exist; Kapralov~\cite{Kapralov13}, building on those techniques, showed
non-existence of one-pass semi-streaming $(1- 1/e+\eps)$-approximation
algorithms for any $\eps > 0$.  A natural next question is: Can we do better in,
say, two passes or three passes?  In answering that, Konrad, Magniez, and
Mathieu~\cite{kmm} gave three-pass and two-pass algorithms that output
matchings that are better than $(1/2)$-approximate.  In this work, we give
algorithms that improve their approximation ratios for two-pass and three-pass
algorithms.  We also give a multi-pass algorithm that does better than the best
known multi-pass algorithms for at least initial few passes.  We are able to
bound the number of passes precisely: we give a $(2/3 -\eps)$-approximation
algorithm that uses $2/(3\eps)$ passes for triangle-free graphs and $4/(3\eps)$
passes for general graphs.  Earlier works either have a large constant inside
the big-$O$ notation for the number of passes~\cite{Eggert2012} or the constant
cannot be determined due to the involved analysis~\cite{mwms,ahnguha}.  For
example, the $(1-\eps)$-approximation algorithm by Eggert et
al.~\cite{Eggert2012} potentially uses $288/\eps^5$ passes, and for the
$(1-\eps)$-approximation algorithms by McGregor~\cite{mwms} and Ahn and
Guha~\cite{ahnguha}, the constants inside the big-$O$ bound cannot be determined
due to the involved analysis.  The $(2/3 -\eps)$-approximation algorithm by
Feigenbaum et al.~\cite{fgnbm} uses $O(\log(1/\eps)/\eps)$ passes, which is
$O(\log(1/\eps))$ factor larger than the number of passes we use to get the same
approximation ratio.  Our algorithms are simple and combinatorial, use
$O(n \log n)$ space, and (can be implemented to) have $O(1)$ update time per
edge.  We also give an explicit and tight analysis of the three-pass algorithm
by Konrad et al.~\cite{kmm} that is reminiscent of Feigenbaum et
al.'s~\cite{fgnbm} multi-pass algorithm.

\paragraph{Technical overview:}

If we can find a matching $M$ such that there are no augmenting paths of length
$3$ in $M \cup M^*$, where $M^*$ is a maximum matching, then $M$ is
$(2/3)$-approximate, i.e., $(1/2 + 1/6)$-approximate.  This is because, in each
connected component of $M\cup M^*$, the ratio of $M$-edges to $M^*$-edges is at
least $2/3$.  This is the basis for the $(2/3 -\eps)$-approximation algorithm by
Feigenbaum et al.~\cite{fgnbm} that uses $O(\log(1/\eps)/\eps)$ passes.  The
same idea is used by Konrad et al.~\cite{kmm} in the analysis of their two-pass
algorithms.  In the first pass, they find a maximal matching $M_0$ and some
subset of support edges, say $S$.  If $M_0$ is so bad that $M_0\cup M^*$ is
almost entirely made up of augmenting paths of length $3$ (i.e.,
$|M_0| \approx |M^*|/2$), then by the end of the second pass, they manage to
augment (using length-$3$ augmentations) a constant fraction of $M_0$ using $S$
and a fresh access to the edges, resulting in a better than
$(1/2)$-approximation. On the other hand, if $M_0$ is not so bad, then they
already have a good matching.  One limitation this idea faces is that a fraction
of the edges in $S$ may become useless for an augmentation if both its endpoints
get matched in $M_0$ by the end of the first pass.  Our main result is a
two-pass algorithm (described in Section~\ref{sec:2pimp2}) that differs in two
ways from the former approach.  Firstly, in the first pass, we only find a
maximal matching $M_0$ so that in the second pass, where we maintain a set $S$
of support edges, $S$ would not contain ``useless'' edges.  Secondly, any
augmentation in our algorithm happens immediately when an edge arrives if it
forms an augmenting path of length $3$ with edges in $M_0$ and $S$.

\paragraph{Our results:} In light of the discussion so far, one way to evaluate
an algorithm is how much advantage it gains over the $(1/2)$-approximate maximal
matching found in the first pass.  We summarize our two-pass and three-pass
results in~\Cref{tab:1} and multi-pass results in~\Cref{tab:2}.  We stress that
we are able to bound the number of passes \emph{precisely}, without big-$O$
notation.
For general graphs, our multi-pass algorithm improves the best known
\emph{deterministic} algorithms in terms of number of passes---see the third
multi-row of~\Cref{tab:2}.  We note that our multi-pass algorithm is not just a
repetition of the second pass of our two-pass algorithm.  Such a repetition will
give an asymptotically worse number of passes (see, for example, the multi-pass
algorithm due to Feigenbaum et al.~\cite{fgnbm}; the first row of~\Cref{tab:2}).  We
carefully choose the parameters for each pass to get the required number of
passes.  Also note that~\Cref{tab:1} shows \emph{advantages} over a maximal
matching---an algorithm is said to have advantage $\alpha$ if it is a
$(1/2 + \alpha)$-approximation algorithm (because a maximal matching is
$(1/2)$-approximate).

\begin{table}[h]
  \caption{Advantages over a maximal matching---advantage $\alpha$ means $(1/2 +
    \alpha)$-approximation.}
  \label{tab:1}
  \begin{center}
    \begin{tabular}{|l|l|l|l|} 
      \hline
      Problem & Previous work & Advantage & Advantage in this work \\ \hline 
      \hline Bipartite two-pass & Esfandiari et al.~\cite{Esf2017}  & $1/12\;\;$
      &\multirow{2}{*}{Not considered separately}\\  
      \cline{1-3} Bipartite three-pass & Esfandiari et al.~\cite{Esf2017}&
      $1/9.52\;\;$ &\\ 
      \hline Triangle-free two-pass & \multicolumn{2}{|c|}{\multirow{2}{*}{Not
          considered separately}}& $1/16\;\;\;\;\;\;\,$ (in
      \Cref{sec:2pimp2})\\
      \cline{1-1} \cline{4-4}Triangle-free three-pass & \multicolumn{2}{|c|}{} &
      $1/10\;\;\;\;\;\;\,$ (in \Cref{sec:3ptf})\\
      \hline General two-pass & Konrad et al.~\cite{kmm} & $1/140\;\;\;\;\;\;\;$ &
      $1/32\;\;\;\;\;\;\,$ (in \Cref{sec:2pimp2})\\ 
      \hline General three-pass & \multicolumn{2}{|c|}{Not considered
        separately} & $1/19.753$ (in \Cref{sec:3pgeneral})\\  
      \hline 
    \end{tabular}
  \end{center}
\end{table}

\begin{table}[h]
  \caption{Multi-pass algorithms---see~\Cref{sec:multi-pass-algor}.}
  \label{tab:2}
  \begin{center}
    \begin{tabular}{|l|l|l|l|} 
      \hline
      Graph & Results & Approx & \# Passes  \\
      \hline \hline
      \multirow{3}{*}{Bipartite}
      & Feigenbaum et al.~\cite{fgnbm} & $2/3 -\eps$ &
      $O(\log(1/\eps)/\eps)$\\\cline{2-4} 
      &Eggert et al.\cite{Eggert2012} &$1-\eps$ & $288/\eps^5$\\\cline{2-4}
      &Ahn and Guha~\cite{ahnguha} &$1-\eps$ &$O(\log\log(1/\eps)/\eps^2)$\\
      \hline
      Triangle free & This work (in \Cref{sec:multi-pass-algor}) & $2/3 -\eps$ &
      $2/(3\eps)$\\\hline 
      \multirow{4}{*}{General}
      & McGregor~\cite{mwms} \;\emph{randomized} & $1 -\eps$ &
      $O((1/\eps)^{1/\eps})$\\\cline{2-4}
      &Ahn and Guha~\cite{ahnguha} &$2/3-\eps$ &
      $O(\log(1/\eps)/\eps^2)$\\\cline{2-4}
      &Ahn and Guha~\cite{ahnguha} &$1-\eps$ &$O(\log n \cdot
      \poly(1/\eps))$\\\cline{2-4} 
      & This work (in \Cref{sec:multi-pass-algor}) & $2/3 -\eps$ & $4/(3\eps)$\\
      \hline
    \end{tabular}
  \end{center}
\end{table}

\paragraph{Note of independent work}\label{par:note}

The work of Esfandiari et al.~\cite{Esf2017} who claim better approximation
ratios for \emph{bipartite graphs} in two passes and three passes is independent
and almost concurrent.  Our work differs in several aspects.  We consider
triangle-free graphs (superset of bipartite graphs) and general graphs, and we
additionally consider multi-pass algorithms.  Also, their algorithm has a
post-processing step that uses time $O(\sqrt{n}\cdot |E|)$, whereas our
algorithms can be implemented to have $O(1)$ update time per edge.  One further
detail about this appears in~\Cref{sec:issue-gener-lemma}.

\subsection{Related Work}
Karp, Vazirani, and Vazirani~\cite{KarpVV90} gave the celebrated
$(1-1/e)$-competitive randomized online algorithm for bipartite graphs in the
vertex arrival setting.  Goel et al.~\cite{gkk} gave the first one-pass
deterministic algorithm with the same approximation ratio, i.e., $1-1/e$, in the
semi-streaming model in the vertex arrival setting.  For the rest of this
section, results involving $\eps$ hold for any $\eps > 0$.  As mentioned
earlier, Goel, Kapralov, and Khanna~\cite{gkk} proved nonexistence of one-pass
$(2/3 + \eps)$-approximation semi-streaming algorithms, which was extended to
$(1 - 1/e +\eps)$-approximation algorithms by Kapralov~\cite{Kapralov13}.  On
the algorithms side, nothing better than outputting a maximal matching, which is
$(1/2)$-approximate, is known.  Closing this gap is considered an outstanding
open problem in the streaming community~\cite{subl60}.

On the multi-pass front, in the semi-streaming model, Feigenbaum et
al.~\cite{fgnbm} gave a $(2/3-\eps)$-approximation algorithm for bipartite
graphs that uses $O(\log(1/\eps)/\eps)$ passes; McGregor~\cite{mwms} improved it
to give a $(1-\eps)$-approximation algorithm for general graphs that uses
$O((1/\eps)^{1/\eps})$ passes. For bipartite graphs, this was again improved by
Eggert et al.~\cite{Eggert2012} who gave a $(1-\eps)$-approximation
$O((1/\eps)^5)$-pass algorithm.  Ahn and Guha~\cite{ahnguha} gave a
linear-programming based $(1-\eps)$-approximation
$O(\log\log(1/\eps)/\eps^2)$-pass algorithm for bipartite graphs.  For general
graphs, their $(1-\eps)$-approximation algorithm uses number of passes
proportional to $\log n$, so it is worse than that of McGregor~\cite{mwms}.

For the problem of one-pass weighted matching, there is a line of work starting
with Feigenbaum et al.~\cite{fgnbm} giving a $6$-approximation semi-streaming
algorithm.  Subsequent results improved this approximation ratio: see
McGregor~\cite{mwms}, Zelke~\cite{zelke}, Epstein et al.~\cite{epsteinlms},
Crouch and Stubbs~\cite{crouchstubbs}, Grigorescu et al.\cite{grigorescuMZ}, and
most recently in a breakthrough, giving a $(2+\eps)$-approximation
semi-streaming algorithm, Paz and Schwartzman~\cite{pazschwa}.  The multi-pass
version of the problem was considered first by McGregor~\cite{mwms}, then by Ahn
and Guha~\cite{ahnguha}.  Chakrabarti and Kale~\cite{ChakraKaleSubmod} and
Chekuri et al.~\cite{cheguqu} consider a more general version of the matching
problem where a submodular function is defined on the edges of the input graph.

The problem of estimating the \emph{size} of a maximum matching (instead of
outputting the actual matching) has also been considered. We mention Kapralov et
al.~\cite{KapralovKhannaSudanMMSE}, Esfandiari et al.~\cite{Esfandiari}, Bury
and Schwiegelshohn~\cite{Bury2015}, and Assadi et al.~\cite{assadiklsoda17}.

In the dynamic streams, edges of the input graph can be removed as well.  The
works of Konrad~\cite{KonradDyn}, Assadi et al.~\cite{AssadiKLY}, and Chitnis et
al.~\cite{Chitnisetl} consider the maximum matching problem in dynamic streams.

\subsection{Organization of the Paper}
\label{sec:organization-paper}
After setting up notation in~\Cref{sec:prelim}, we give a tight analysis of the
three-pass algorithm for bipartite graphs by Konrad et al.~\cite{kmm}
in~\Cref{sec:3pkmm}.  In~\Cref{sec:simple-two-pass}, we see our simple two-pass
algorithm for triangle-free graphs.  Then in~\Cref{sec:2pimp2}, we see our main
result---the improved two-pass algorithm, and then we see the multi-pass
algorithm in~\Cref{sec:multi-pass-algor}.  The results that are not covered in
the main sections are covered in the appendix.


\section{Preliminaries}
\label{sec:prelim}
We work on graph \emph{streams}.  The input is a sequence of edges (stream) of a
graph $G=(V,E)$, where $V$ is the set of vertices and $E$ is the set of edges; a
bipartite graph is denoted as $G = (A,B,E)$.  A streaming algorithm may go over
the stream a few times (multi-pass) and use space $O(n \polylog n)$, where
$n = |V|$.  In this paper, we give algorithms that make two, three, or a few
more passes over the input graph stream.  A matching $M$ is a subset of edges
such that each vertex has at most one edge in $M$ incident to it.  The maximum
cardinality matching problem, or maximum matching, for short, is to find a
largest matching in the given graph.  Our goal is to design streaming algorithms
for maximum matching.

For a subset $F$ of edges and a subset $U$ of vertices, we denote by
$U(F) \subseteq U$ the set of vertices in $U$ that have an edge in $F$ incident
on them.  Conversely, we denote by $F(U)\subseteq F$ the set of edges in $F$
that have an endpoint in $U$.  For a subset $F$ of edges and a vertex
$v\in V(F)$, we denote by $N_F(v)$ the set of $v$'s neighbors in the graph
$(V(F),F)$, and we define $\deg_F(v) := |N_F(v)|$.

In the first pass, our algorithms compute a \emph{maximal} matching which we
denote by $M_0$.  We use $M^*$ to indicate a matching of maximum cardinality.
Assume that $M_0$ and $M^*$ are given.  For $i \in \{3,5,7,\ldots\}$, a
connected component of $M_0\cup M^*$ that is a path of length $i$ is called an
$i$-augmenting path (nonaugmenting otherwise).  We say that an edge in $M_0$ is
$3$-augmentable if it belongs to a $3$-augmenting path, otherwise we say that it
is non-$3$-augmentable.

\begin{lem}[Lemma 1 in~\cite{kmm}]
  \label{lem:kmmlem1}
  Let $\advntg \ge 0$, $M_0$ be a maximal matching in $G$, and $M^*$ be a maximum
  matching in $G$ such that $|M_0| \le (1/2 + \advntg) |M^*|$.  Then the number of
  $3$-augmentable edges in $M_0$ is at least $(1/2 - 3\advntg)|M^*|$, and the
  number of non-$3$-augmentable edges in $M_0$ is at most $4\advntg|M^*|$.
\end{lem}
\begin{proof}
  Let the number of $3$-augmentable edges in $M_0$ be $k$. For each
  $3$-augmentable edge in $M_0$, there are two edges in $M^*$ incident on it.
  Also, each non-$3$-augmentable edge in $M_0$ lies in a connected component of
  $M_0 \cup M^*$ in which the ratio of the number of $M^*$-edges to the number
  of $M_0$-edges is at most $3/2$.  Hence, \allowdisplaybreaks
  \begin{align*}
    |M^*| &\le 2k + \frac{3}{2} (|M_0| - k)
    &&\text{ since there are $|M_0|-k$ non-$3$-augmentable edges}\,,\\
          &\le 2k + \frac{3}{2} \left(\left(\frac{1}{2} + \advntg \right)|M^*| -
            k\right)
    &&\text{ because $|M_0| \le (1/2 + \advntg) |M^*|$}\,,\\
          &= \frac{1}{2}k + \left(\frac{3}{4} + \frac{3}{2}\advntg
            \right)|M^*|\,,
  \end{align*}
  which, after simplification, gives $k \ge (1/2 - 3\advntg) |M^*|$.  And the
  number of non-$3$-augmentable edges in $M_0$ is
  $|M_0| - k \le |M_0| - (1/2 - 3\advntg)|M^*| \le (1/2 +\advntg - 1/2 +
  3\advntg)|M^*| = 4\advntg|M^*|$.
\end{proof}
We make the following simple, yet crucial, observation.
\begin{note}
  Let $M_0$ be a maximal matching.  Then $V(M_0)$ is a vertex cover, and there
  is no edge between any two vertices in $V\setminus V(M_0)$.  Therefore, even
  if the input graph is not a bipartite graph, the set of edges incident on
  $V\setminus V(M_0)$, i.e., $E(V\setminus V(M_0))$ give rise to a bipartite
  graph with bipartition $(V\setminus V(M_0), V(M_0))$.
\end{note}

For all the algorithms in this paper, it can be verified that their space
complexity is $O(n\log n)$ and update time per edge is $O(1)$.  We also ignore
floors and ceilings for the sake of exposition.

\section{Analyzing the Three Pass Algorithm for Bipartite
  Graphs}\label{sec:3pkmm}
We analyze the three-pass algorithm for bipartite graphs given by Konrad et
al.~\cite{kmm}, i.e.,~\Cref{alg:3pkmm} by considering the distribution of
lengths of augmenting paths.  We also give a tight example.
\algrenewcommand\algorithmicforall{\textbf{foreach}}
\begin{algorithm}[!ht]
  \caption{~~Three-pass algorithm for bipartite graphs due to Konrad et
    al.~\cite{kmm} \label{alg:3pkmm}}
  \begin{algorithmic}[1]
    \State In the first pass, find a maximal matching $M_0$.
    \State In the second pass, find a maximal matching
    \begin{itemize}
    \item $M_A$ in $F_2 := \{ab : a \in A(M_0), b \in B\setminus B(M_0)\}$
    (see~\Cref{fig:3pkmm}). 
    \end{itemize}
    \State In the third pass, find a maximal matching
    \begin{itemize}
    \item $M_B$ in
      $F_3 := \{ab : a \in A \setminus A(M_0) \text{ and } \exists a' \in A(M_A)
      \text{ such that } a'b \in M_0\}$.
    \end{itemize}
    \State Augment $M_0$ using edges in $M_A$ and $M_B$ and return the resulting
    matching $M$.
  \end{algorithmic}
\end{algorithm}
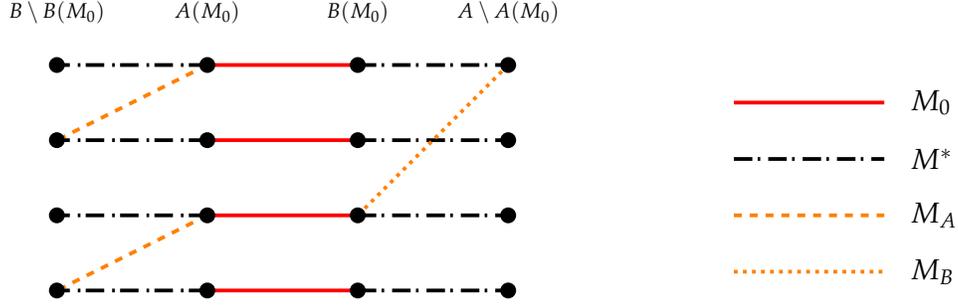
\begin{figure}
\centering
\begin{tikzpicture}[line width=0.5mm]

\draw [red] (0,0) -- (2,0);
\draw [red] (0,1) -- (2,1);
\draw [red] (0,2) -- (2,2);
\draw [red] (0,3) -- (2,3);

\draw [dash pattern={on 7pt off 2pt on 1pt off 3pt}] (2,0) -- (4,0);
\draw [dash pattern={on 7pt off 2pt on 1pt off 3pt}] (2,1) -- (4,1);
\draw [dash pattern={on 7pt off 2pt on 1pt off 3pt}] (2,2) -- (4,2);
\draw [dash pattern={on 7pt off 2pt on 1pt off 3pt}] (2,3) -- (4,3);

\draw [dash pattern={on 7pt off 2pt on 1pt off 3pt}] (0,0) -- (-2,0);
\draw [dash pattern={on 7pt off 2pt on 1pt off 3pt}] (0,1) -- (-2,1);
\draw [dash pattern={on 7pt off 2pt on 1pt off 3pt}] (0,2) -- (-2,2);
\draw [dash pattern={on 7pt off 2pt on 1pt off 3pt}] (0,3) -- (-2,3);

\draw [dashed,orange] (0,1) -- (-2,0);
\draw [dashed,orange] (0,3) -- (-2,2);

\draw [dotted,orange] (4,3) -- (2,1);

\draw [red] (7,2.5) -- (9,2.5);
\draw [dash pattern={on 7pt off 2pt on 1pt off 3pt}] (7,1.75) -- (9,1.75);
\draw [dashed, orange] (7,1) -- (9,1);
\draw [dotted,orange] (7,0.25) -- (9,0.25);

\draw (10.5,2.5) node[anchor=east,text width=1cm] {$M_0$};
\draw (10.5,1.75) node[anchor=east,text width=1cm] {$M^*$};
\draw (10.5,1) node[anchor=east,text width=1cm] {$M_A$};
\draw (10.5,0.25) node[anchor=east,text width=1cm] {$M_B$};

\draw (0,4) node[anchor=north] {\scriptsize $A(M_0)$};
\draw (-2,4) node[anchor=north] {\scriptsize $B\setminus B(M_0)$};
\draw (2,4) node[anchor=north] {\scriptsize $B(M_0)$};
\draw (4,4) node[anchor=north] {\scriptsize $A\setminus A(M_0)$};

\fill [color=black] (0,0) circle (3pt);
\fill [color=black] (2,0) circle (3pt);
\fill [color=black] (0,1) circle (3pt);
\fill [color=black] (2,1) circle (3pt);
\fill [color=black] (0,2) circle (3pt);
\fill [color=black] (2,2) circle (3pt);
\fill [color=black] (0,3) circle (3pt);
\fill [color=black] (2,3) circle (3pt);

\fill [color=black] (-2,0) circle (3pt);
\fill [color=black] (4,0) circle (3pt);
\fill [color=black] (-2,1) circle (3pt);
\fill [color=black] (4,1) circle (3pt);
\fill [color=black] (-2,2) circle (3pt);
\fill [color=black] (4,2) circle (3pt);
\fill [color=black] (-2,3) circle (3pt);
\fill [color=black] (4,3) circle (3pt);

\end{tikzpicture}
\caption{Example: state of variables in an execution of
  \Cref{alg:3pkmm}.} \label{fig:3pkmm}
\end{figure}

\begin{theorem}
  \label{thm:3pkmmimpimp}
  \Cref{alg:3pkmm} is a three-pass, semi-streaming,
  $(1/2 +1/10)$-approximation algorithm for maximum matching in bipartite
  graphs.  
\end{theorem}
\begin{proof}
  Without loss of generality, let $M^*$ be a maximum matching such that all
  nonaugmenting connected components of $M_0 \cup M^*$ are single edges.  For
  $i = \{3,5,7,\ldots\}$, let $k_i$ denote the number of $i$-augmenting paths in
  $M_0 \cup M^*$, and let $k = |M_0 \cap M^*|$.  Then
  \begin{equation}
    \label{eq:m0mstarval}
    |M_0| = k + \sum_i \frac{i-1}{2}k_i \;\;\;\text{ and }\;\;\;
    |M^*| = k + \sum_i \frac{i+1}{2}k_i\,.  
  \end{equation}
  Consider an $i$-augmenting path $b_1a_1b_2a_2b_3\cdots b_{(i+1)/2}a_{(i+1)/2}$
  in $M_0 \cup M^*$, where for each $j$, we have $a_j \in A$ and $b_j \in B$.
  We call the vertex $a_{(i-1)/2}$ a good vertex, because an edge in $M_A$
  incident to $a_{(i-1)/2}$ can potentially be augmented using the edge
  $b_{(i+1)/2}a_{(i+1)/2}$.  To elaborate, consider the set of all edges in
  $M_A$ incident on good vertices; call it $M'_A$.  Consider the set of edges of
  the type $b_{(i+1)/2}a_{(i+1)/2}$ from each $i$-augmenting path; call it
  $M_F$.  Note that $M_F$ is a matching.  Then we can augment $M_0$ using $M'_A$
  and $M_F$ by as much as $|M'_A|$.

  There is a matching of size $\sum_i k_i$ in $F_2$ formed by edges of the type
  $b_1a_1$ from each $i$-augmenting path.  Since $M_A$ is maximal in $F_2$, we
  have $|M_A| \ge (\sum_i k_i)/2$.  Now, the number of good vertices is
  $\sum_i k_i$; therefore, the number of bad (i.e., not good) vertices is
  $|M_0| - \sum_i k_i$.  So the number of edges in $M_A$ incident on good
  vertices (see~\Cref{fig:3pkmmte})
  \[
    |M'_A|
    \ge \frac{\sum_i k_i}{2} -  \left(|M_0| - \sum_i k_i\right)
    = \frac{3}{2}\sum_i k_i - |M_0|\,.
  \]
  Let
  $B_G := \{ b \in B: \exists a \in A(M'_A) \text{ such that } ab \in M_0\}$.
  Let $M'_F \subseteq M_F$ be defined as $M'_F := \{ba \in M_F : b \in B_G\}$.
  Then we know that $|M'_F| =|M'_A|$ and $M'_F\subseteq M_F \subseteq F_3$.
  Since we select a maximal matching in $F_3$ in the third pass,
  \begin{equation}
    \label{eq:mb5apbound}
    |M_B| \ge \frac{|M'_F|}{2} = \frac{|M'_A|}{2}
    =\frac{3}{4}\sum_i k_i - \frac{|M_0|}{2}\,.
  \end{equation}
  So the output size
  \begin{align*}
    |M| &= |M_0| + |M_B|\\
        &\ge |M_0| + \frac{3}{4}\sum_i k_i - \frac{|M_0|}{2}
        &&\text{ by~\eqref{eq:m0mstarval}~and~\eqref{eq:mb5apbound},}\\
        &= \frac{|M_0|}{2} + \frac{3}{4}(|M^*| - |M_0|)
        &&\text{ by~\eqref{eq:m0mstarval}, $\sum_i k_i = |M^*| - |M_0|$}\,,
  \end{align*}
  i.e., $|M| \ge 3|M^*|/4 - |M_0|/4$, but we also have $|M|\ge |M_0|$, hence
  \[
    |M| \ge \max\left\{|M_0|, \frac{3}{4}|M^*| - \frac{1}{4}|M_0|\right\}\,.
  \]
  So the bound is minimized when
  $|M_0| = 3|M^*|/4 - |M_0|/4 = 3|M^*|/5 = (1/2 + 1/10)|M^*|$.
\end{proof}
As we can see in the proof above, the worst case happens when
$|M| = |M_0| = 3|M^*|/5$.  Setting $k_3 = k_5 \ge 1$, $k=0$, and $k_i=0$ for
$i > 5$ gives us the tight example shown in~\Cref{fig:3pkmmte}.
\begin{figure}
\centering
\begin{tikzpicture}[line width=0.5mm]

\draw [red] (0,0) -- (2,0);
\draw [red] (0,1) -- (2,1);
\draw [red] (0,2) -- (2,2);

\draw [dash pattern={on 7pt off 2pt on 1pt off 3pt}] (2,0) -- (4,0);
\draw [dash pattern={on 7pt off 2pt on 1pt off 3pt}] (2,1) -- (4,1);
\draw [dash pattern={on 7pt off 2pt on 1pt off 3pt}] (-2,2) -- (0,2);
\draw [dash pattern={on 7pt off 2pt on 1pt off 3pt}] (-2,0) -- (0,0);
\draw [dash pattern={on 7pt off 2pt on 1pt off 3pt}] (0,1) -- (2,2);

\draw [dashed,orange] (0,2) -- (-2,0);

\draw [red] (6,1.85) -- (8,1.85);
\draw [dash pattern={on 7pt off 2pt on 1pt off 3pt}] (6,1.1) -- (8,1.1);
\draw [dashed, orange] (6,0.35) -- (8,0.35);

\draw (9.5,1.85) node[anchor=east,text width=1cm] {$M_0$};
\draw (9.5,1.1) node[anchor=east,text width=1cm] {$M^*$};
\draw (9.5,0.35) node[anchor=east,text width=1cm] {$M_A$};

\draw (0,3) node[anchor=north] {\scriptsize $A(M_0)$};
\draw (-2,3) node[anchor=north] {\scriptsize $B\setminus B(M_0)$};
\draw (2,3) node[anchor=north] {\scriptsize $B(M_0)$};
\draw (4,3) node[anchor=north] {\scriptsize $A\setminus A(M_0)$};

\fill [color=black] (0,0) circle (3pt);
\fill [color=black] (2,0) circle (3pt);
\fill [color=black] (0,1) circle (3pt);
\fill [color=black] (2,1) circle (3pt);
\fill [color=black] (0,2) circle (3pt);
\fill [color=black] (2,2) circle (3pt);

\fill [color=black] (-2,0) circle (3pt);
\fill [color=black] (4,1) circle (3pt);
\fill [color=black] (-2,2) circle (3pt);
\fill [color=black] (4,0) circle (3pt);

\end{tikzpicture}
\caption{Tight example for~\Cref{alg:3pkmm}: $M_A$ has only one edge that
    lands on a bad vertex and cannot be augmented in the third pass. So $|M| =
    |M_0| = 3$ and $|M^*| = 5$.}\label{fig:3pkmmte}
\end{figure}



\section{A Simple Two Pass Algorithm for Triangle Free Graphs}
\label{sec:simple-two-pass}
Before seeing our main result, we see a simple two pass algorithm for
triangle-free graphs. The function \textproc{Semi}$()$ in~\Cref{alg:2ptf}
greedily computes a subset of edges such that each vertex in $X$ has degree at
most one and each vertex in $Y$ has degree at most $\lambda$; we call such a
subset a $(\lambda,X,Y)$-semi-matching (Konrad et al.~\cite{kmm} call this a
$\lambda$-bounded semi-matching).  In~\Cref{alg:2ptf}, we find a maximal
matching $M_0$ in the first pass, and, in the second pass, we find a
$(\lambda,V(M_0),V\setminus V(M_0))$-semi-matching $S$.  After the second pass,
we greedily augment edges in $M_0$ one by one using edges in $S$.

\algrenewcommand\algorithmicforall{\textbf{foreach}}
\begin{algorithm}[!ht]
  \caption{~Two-pass algorithm for triangle-free graphs \label{alg:2ptf}}
  \begin{algorithmic}[1]
    \State In the first pass: $M_0\gets$ maximal matching
    \State In the second pass: $S\gets\textproc{Semi}(\lambda,V(M_0),V\setminus
    V(M_0))$ (see~\Cref{fig:2ptfex}).
    \State After the second pass, augment $M_0$ greedily using edges in $S$ to
    get $M$; output $M$.
    \Function{Semi}{$\lambda,X,Y$}\Comment{based on
      Algorithm 7 in Konrad et al.~\cite{kmm}}
    \State $S\gets\emptyset$
    \ForAll{edge $xy$ such that $x\in X$, $y\in Y$}
    \If{$\deg_S(x) = 0$ and $\deg_S(y) \le \lambda - 1$}
    \State $S\gets S\cup \{xy\}$
    \EndIf
    \EndFor
    \EndFunction
  \end{algorithmic}
\end{algorithm}
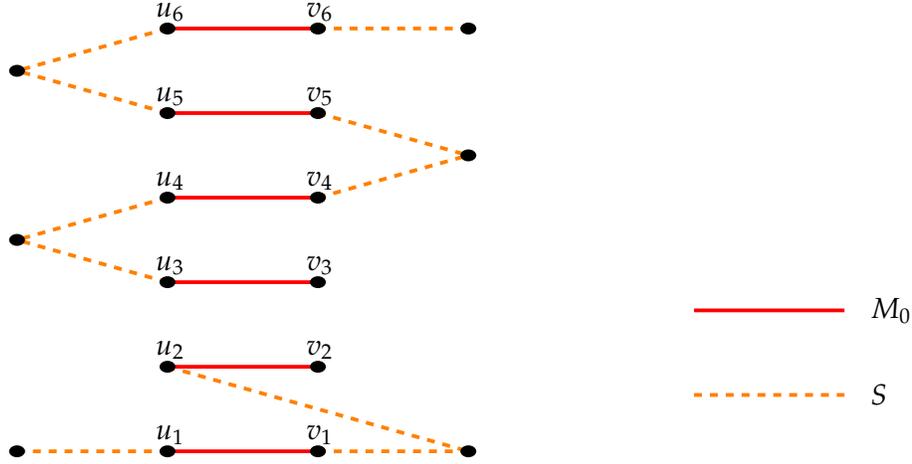
\begin{figure}
\centering
\begin{tikzpicture}[yscale=0.75,line width=0.5mm]
\draw [red] (0,0) -- (2,0);
\draw [red] (0,1.5) -- (2,1.5);
\draw [red] (0,3) -- (2,3);
\draw [red] (0,4.5) -- (2,4.5);
\draw [red] (0,6) -- (2,6);
\draw [red] (0,7.5) -- (2,7.5);

\draw [dashed,orange] (0,7.5) -- (-2,6.75);
\draw [dashed,orange] (0,6) -- (-2,6.75);
\draw [dashed,orange] (0,4.5) -- (-2,3.75);
\draw [dashed,orange] (0,3) -- (-2,3.75);
\draw [dashed,orange] (2,4.5) -- (4,5.25);
\draw [dashed,orange] (2,6) -- (4,5.25);
\draw [dashed,orange] (2,0) -- (4,0);
\draw [dashed,orange] (0,1.5) -- (4,0);
\draw [dashed,orange] (2,7.5) -- (4,7.5);
\draw [dashed,orange] (-2,0) -- (0,0);

\draw [red] (7,2.5) -- (9,2.5);
\draw [dashed, orange] (7,1) -- (9,1);

\draw (10.5,2.5) node[anchor=east,text width=1cm] {$M_0$};
\draw (10.5,1) node[anchor=east,text width=1cm] {$S$};
\draw (1,0.3) node[anchor=east,text width=1cm] {$u_1$};
\draw (3,0.3) node[anchor=east,text width=1cm] {$v_1$};
\draw (1,1.8) node[anchor=east,text width=1cm] {$u_2$};
\draw (3,1.8) node[anchor=east,text width=1cm] {$v_2$};
\draw (1,3.3) node[anchor=east,text width=1cm] {$u_3$};
\draw (3,3.3) node[anchor=east,text width=1cm] {$v_3$};
\draw (1,4.8) node[anchor=east,text width=1cm] {$u_4$};
\draw (3,4.8) node[anchor=east,,text width=1cm] {$v_4$};
\draw (1,6.3) node[anchor=east,text width=1cm] {$u_5$};
\draw (3,6.3) node[anchor=east,text width=1cm] {$v_5$};
\draw (1,7.8) node[anchor=east,text width=1cm] {$u_6$};
\draw (3,7.8) node[anchor=east,text width=1cm] {$v_6$};

\fill [color=black] (0,0) circle (3pt); 
\fill [color=black] (2,0) circle (3pt); 
\fill [color=black] (0,1.5) circle (3pt);
\fill [color=black] (2,1.5) circle (3pt);
\fill [color=black] (0,3) circle (3pt);
\fill [color=black] (2,3) circle (3pt);
\fill [color=black] (0,4.5) circle (3pt);
\fill [color=black] (2,4.5) circle (3pt);
\fill [color=black] (0,6) circle (3pt);
\fill [color=black] (2,6) circle (3pt);
\fill [color=black] (0,7.5) circle (3pt);
\fill [color=black] (2,7.5) circle (3pt);

\fill [color=black] (-2,6.75) circle (3pt);
\fill [color=black] (-2,3.75) circle (3pt);
\fill [color=black] (-2,0) circle (3pt);
\fill [color=black] (4,5.25) circle (3pt);
\fill [color=black] (4,0) circle (3pt);
\fill [color=black] (4,7.5) circle (3pt);
\end{tikzpicture}
\caption{Example showing $M_0$ and $S$ at the end of the second pass of
  \Cref{alg:2ptf} with $\lambda = 2$.  When we greedily augment $M_0$ after the
  second pass, we may choose to augment $u_5v_5$ and lose two possible
  augmentations of edges $u_4v_4$ and $u_6v_6$.} \label{fig:2ptfex}
\end{figure}

\begin{theorem}
  \label{thm:2ptf}
  \Cref{alg:2ptf} is a two-pass, semi-streaming,
  $(1/2 +1/20)$-approximation algorithm for maximum matching in
  triangle-free graphs.
\end{theorem}
\begin{proof}
  As in the proof of~\Cref{thm:3pkmmimpimp}, let $M^*$ be a maximum matching
  such that all nonaugmenting connected components of $M_0 \cup M^*$ are single
  edges.  For $i = \{3,5,7,\ldots\}$, let $k_i$ denote the number of
  $i$-augmenting paths in $M_0 \cup M^*$, and let $k$ denote the number of edges
  in $M^* \cap M_0$.

  Now, we define \emph{good} vertices.  Consider an $i$-augmenting path
  $x_1y_1x_2y_2x_3\cdots x_{(i+1)/2}y_{(i+1)/2}$ in $M_0 \cup M^*$.  We call the
  vertices $y_1\in V(M_0)$ and $x_{(i+1)/2}\in V(M_0)$ good vertices, because
  the edges $x_1y_1 \in M^*$ and $x_{(i+1)/2}y_{(i+1)/2} \in M^*$ can
  potentially be added to $S$ by our algorithm.  Denote by $V_G$ the set of good
  vertices and by $V_B := V(M_0)\setminus V_G$ the set of bad vertices.  Then
  $|V_G| = 2\sum_i k_i$.  Note that $V_G\cap V_B =\emptyset$ and
  $V_G\cup V_B = V(M_0)$ by definition.

  Let $\Vnc := V_G\setminus V(S)$ be the set of good vertices \emph{not covered}
  by $S$.  An edge $uv \in M^*$ with $u\in V\setminus V(M_0)$ and $v\in \Vnc$
  was not added to $S$, because $\deg_S(u) = \lambda$. Hence
  \begin{equation}
    \label{eq:4}
    \lambda |\Vnc| \le |V(M_0)| - |\Vnc|\,\;\;\;\text{ i.e., }\;\;\; |\Vnc| \le
    \frac{2}{\lambda+1}|M_0|\,, 
  \end{equation}
  because at most $|V(M_0)| - |\Vnc|$ vertices in $V(M_0)$ are covered by $S$.
  Now,
  \begin{align*}
    |V(M_0)\setminus V(S)|
    &= |V_G\setminus V(S)| + |V_B\setminus V(S)|
    &&\because \text{ $V_G\cap V_B =\emptyset$ and $V_G\cup V_B = V(M_0)$,}\\ 
    &\le|\Vnc| + |V_B|
    &&\because \text{ $\Vnc = V_G\setminus V(S)$ and $|V_B\setminus V(S)| \le
       |V_B|$,}\\
    &\le\frac{2}{\lambda+1}|M_0| + |V(M_0)| - |V_G|
    &&\text{by~\eqref{eq:4} and the fact $|V_B| = |V(M_0)| - |V_G|$,}\\
    &=\frac{2}{\lambda+1}|M_0| + |V(M_0)| - 2\sum_i k_i
    &&\text{because $|V_G| = 2\sum_i k_i$}\,.
  \end{align*}
  Using $|V(M_0)| = |V(M_0)\setminus V(S)| + |V(M_0)\cap V(S)|$ and the above,
  we get
  \begin{align}
    |V(M_0)\cap V(S)|
    &\ge |V(M_0)| - \left( \frac{2}{\lambda+1}|M_0| + |V(M_0)| - 2\sum_i k_i
      \right)\nonumber\\
    &=2\left(\sum_i k_i - \frac{1}{\lambda+1}|M_0|\right)\,.\label{eq:4a}
  \end{align}
  We observe that at most $|M_0|$ vertices in $V(M_0)$ (one endpoint of each
  edge) can be covered by $S$ without having both endpoints of an edge in $M_0$
  covered.  Hence, at least $|V(M_0)\cap V(S)| - |M_0|$ edges in $M_0$ have both
  their endpoints covered by $S$, which, by~\eqref{eq:4a}, is at least
  \begin{equation}
    \label{eq:5}
    2\left(\sum_i k_i - \frac{1}{\lambda+1}|M_0|\right) - |M_0|
    =2\sum_i k_i - \frac{\lambda+3}{\lambda+1}|M_0|\,.
  \end{equation}
  After the second pass, when we greedily augment an edge from the above edges,
  i.e., edges whose both endpoints are covered by $S$, we may potentially lose
  $2(\lambda-1)$ other augmentations (see~\Cref{fig:2ptfex}).  To see this,
  consider $uv\in M_0$ such that $u,v \in V(S)$ and $au \in S$ and $vb \in S$.
  The graph is triangle free, so we know that $a \neq b$, and we \emph{can}
  augment $M_0$ using the $3$-augmenting path $auvb$; but we may lose at most
  $\lambda-1$ edges incident to $a$ in $S$ and at most $\lambda-1$ edges
  incident to $b$ in $S$.  Therefore the number of augmentations $c$ we get
  after the second pass is at least $1/(2\lambda-1)$ times the right hand side
  of~\eqref{eq:5}, i.e.,
  \[
    c \ge \frac{2}{2\lambda-1}\sum_i k_i -
    \frac{\lambda+3}{(2\lambda-1)(\lambda+1)}|M_0|\,.
  \]
  So the output size $|M| = |M_0| + c$, and using the above bound on $c$ and
  simplifying we get:
  \[
    |M| \ge \frac{2}{2\lambda-1}\sum_i k_i +
    \frac{2(\lambda^2-2)}{(2\lambda-1)(\lambda+1)}|M_0|\,;
  \]
  substituting $\sum_i k_i = |M^*|-|M_0|$, by~\eqref{eq:m0mstarval}, in the
  above,
  \[
    |M| \ge  \frac{2}{2\lambda-1}|M^*| +
    \frac{2(\lambda^2 -\lambda-3)}{(2\lambda-1)(\lambda+1)}|M_0|\,.
  \]
  Using $\lambda=3$ and the fact that $M_0$ is $2$-approximate, we get
  \[
    |M| \ge \frac{2}{5}|M^*| + \frac{3}{10}|M_0|
    \ge \frac{2}{5}|M^*| + \frac{3}{20}|M^*|
    = \frac{11}{20}|M^*|
    =\left(\frac{1}{2}+ \frac{1}{20}\right)|M^*|\,.\qedhere
  \]    
\end{proof}


\section{Improved Two Pass Algorithm}\label{sec:2pimp2}
We present our main result that is a two pass algorithm in this section.  In the
first pass, we find a maximal matching $M_0$.  In the second pass, we maintain a
set $S$ of support edges $xy$, such that $x\in V\setminus V(M_0)$,
$y\in V(M_0)$, and $\deg_S(y)\leq \lambdaM$ and $\deg_S(x)\leq \lambdaU$, where
$\lambdaM \ge 1$ and $\lambdaU \ge 1$ are parameters denoting maximum degree
allowed in $S$ for matched and unmatched vertices (with respect to $M_0$),
respectively.  Whenever a new edge forms a $3$-augmenting path with an edge in
$M_0$ and an edge in $S$, we augment.  \Cref{alg:2pimp2} gives a formal
description.
\begin{algorithm}[!ht]
  \caption{Improved two-pass algorithm\label{alg:2pimp2}: input graph $G$}
  \begin{algorithmic}[1]
    \State In the first pass, find a maximal matching $M_0$.
    \If{$G$ is triangle-free}
    \State Return \textproc{Improve-Matching}$(M_0,2,1)$
    \Else
    \State Return \textproc{Improve-Matching}$(M_0,4,2)$
    \EndIf

    \Function{Improve-Matching}{$M_0, \lambdaU, \lambdaM$}
    \label{alg:2pimp2:fimpm}
    \State  $M\gets M_0$, $S\gets\emptyset$, $I\gets \emptyset$ and
    $I_B\gets \emptyset$
    \ForAll{edge $xy$ in the stream} \label{line:alg:2pimp2:4}
    \If{$x \text{ or } y\in I \cup I_B$}\label{line:9}
    \State Continue, i.e., ignore $xy$.
    \ElsIf{$x\in V(M_0)$ and $y\in V(M_0)$}
    \State Continue, i.e., ignore $xy$.
    \ElsIf{there exist $v$ and $b$ such that $yv \in M_0$ and $vb \in S$}
    \State
    $M\gets M\setminus \{yv\} \cup \{xy,vb\}$ \label{line:c3aug}
    \Comment{a $3$-augmentation}
    \Statex \hskip\algorithmicindent \hskip\algorithmicindent
    \hskip\algorithmicindent 
    Let
    $I_x\gets\{u_x\,,v_x: xu_x\in S \text{ and } u_xv_x\in M_0 \}$.
    \Statex \hskip\algorithmicindent \hskip\algorithmicindent
    \hskip\algorithmicindent 
    Let
    $I_b\gets\{u_b\,,v_b: u_bv_b\in M_0 \text{ and } v_bb\in S \}$. 
    \State Then $I\gets I\cup \{x,y,v,b\}$ and $I_B\gets I_B \cup I_x \cup I_b$.
    \Else
    \Statex \hskip\algorithmicindent \hskip\algorithmicindent
    \hskip\algorithmicindent 
    Without loss of generality, assume that $x\in V\setminus V(M_0)$ and $y\in
    V(M_0)$.
    \If{$\deg_S(x) <\lambdaU$ and $\deg_S(y)<\lambdaM$}
    \Comment{See~\Cref{fig:2pimp2}.}
    \State $S\gets S\cup \{xy\}$
    \Comment{\textbf{Note:} Once an edge is added to $S$, it is never removed
      from it.}
    \EndIf
    \EndIf
    \EndFor
    \State Return $M$.
    \EndFunction
  \end{algorithmic}
\end{algorithm}
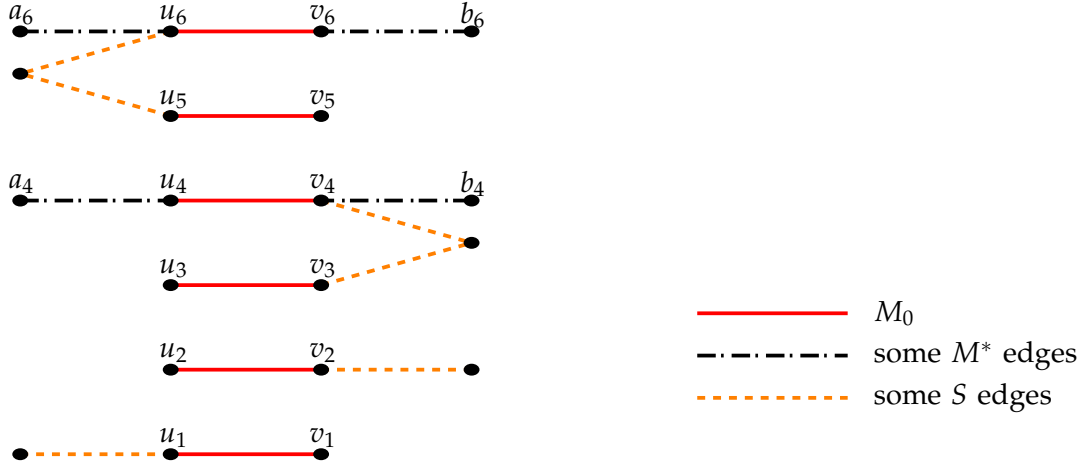
\begin{figure}
\centering
\begin{tikzpicture}[yscale=0.75,line width=0.5mm]
\draw [red] (0,0) -- (2,0);
\draw [red] (0,1.5) -- (2,1.5);
\draw [red] (0,3) -- (2,3);
\draw [red] (0,4.5) -- (2,4.5);
\draw [red] (0,6) -- (2,6);
\draw [red] (0,7.5) -- (2,7.5);

\draw [dashed,orange] (0,7.5) -- (-2,6.75);
\draw [dashed,orange] (0,6) -- (-2,6.75);
\draw [dashed,orange] (2,4.5) -- (4,3.75);
\draw [dashed,orange] (2,3) -- (4,3.75);
\draw [dashed,orange] (-2,0) -- (0,0);
\draw [dashed,orange] (2,1.5) -- (4,1.5);

\draw [dash pattern={on 7pt off 2pt on 1pt off 3pt}] (-2,7.5) -- (0,7.5);
\draw [dash pattern={on 7pt off 2pt on 1pt off 3pt}] (4,7.5) -- (2,7.5);
\draw [dash pattern={on 7pt off 2pt on 1pt off 3pt}] (-2,4.5) -- (0,4.5);
\draw [dash pattern={on 7pt off 2pt on 1pt off 3pt}] (4,4.5) -- (2,4.5);

\draw [red] (7,2.5) -- (9,2.5);
\draw [dash pattern={on 7pt off 2pt on 1pt off 3pt}] (7,1.75) -- (9,1.75);
\draw [dashed, orange] (7,1) -- (9,1);

\draw (10.5,2.5) node[anchor=east,text width=1cm] {$M_0$};
\draw (12.5,1.75) node[anchor=east,text width=3cm] {some $M^*$ edges};

\draw (12.5,1) node[anchor=east,text width=3cm] {some $S$ edges};
\draw (1,0.3) node[anchor=east,text width=1cm] {$u_1$};
\draw (3,0.3) node[anchor=east,text width=1cm] {$v_1$};
\draw (1,1.8) node[anchor=east,text width=1cm] {$u_2$};
\draw (3,1.8) node[anchor=east,text width=1cm] {$v_2$};
\draw (1,3.3) node[anchor=east,text width=1cm] {$u_3$};
\draw (3,3.3) node[anchor=east,text width=1cm] {$v_3$};
\draw (1,4.8) node[anchor=east,text width=1cm] {$u_4$};
\draw (3,4.8) node[anchor=east,,text width=1cm] {$v_4$};
\draw (1,6.3) node[anchor=east,text width=1cm] {$u_5$};
\draw (3,6.3) node[anchor=east,text width=1cm] {$v_5$};
\draw (1,7.8) node[anchor=east,text width=1cm] {$u_6$};
\draw (3,7.8) node[anchor=east,text width=1cm] {$v_6$};
\draw (-1,7.8) node[anchor=east,text width=1cm] {$a_6$};
\draw (5,7.8) node[anchor=east,text width=1cm] {$b_6$};
\draw (-1,4.8) node[anchor=east,text width=1cm] {$a_4$};
\draw (5,4.8) node[anchor=east,text width=1cm] {$b_4$};

\fill [color=black] (0,0) circle (3pt); 
\fill [color=black] (2,0) circle (3pt); 
\fill [color=black] (0,1.5) circle (3pt);
\fill [color=black] (2,1.5) circle (3pt);
\fill [color=black] (0,3) circle (3pt);
\fill [color=black] (2,3) circle (3pt);
\fill [color=black] (0,4.5) circle (3pt);
\fill [color=black] (2,4.5) circle (3pt);
\fill [color=black] (0,6) circle (3pt);
\fill [color=black] (2,6) circle (3pt);
\fill [color=black] (0,7.5) circle (3pt);
\fill [color=black] (2,7.5) circle (3pt);

\fill [color=black] (-2,7.5) circle (3pt);
\fill [color=black] (4,7.5) circle (3pt);
\fill [color=black] (-2,4.5) circle (3pt);
\fill [color=black] (4,4.5) circle (3pt);
\fill [color=black] (-2,6.75) circle (3pt);
\fill [color=black] (4,3.75) circle (3pt);
\fill [color=black] (4,1.5) circle (3pt);
\fill [color=black] (-2,0) circle (3pt);
\end{tikzpicture}
\caption{Example showing $M_0$ and some of the edges in $M^*$ and $S$ during the
  second pass of \Cref{alg:2pimp2} for triangle-free graphs with $\lambdaU = 2$
  and $\lambdaM = 1$.  At most one of $u_i$ and $v_i$ can have positive degree
  in $S$, because we would rather augment $u_iv_i$ instead of adding the latter
  edge to $S$.  By our convention, $a_4u_4$ arrived before $v_4b_4$, and
  $a_6u_6$ arrived before $v_6b_6$. Since $a_4u_4$ was not added to $S$, we have
  $\deg_S(a_4) = \lambdaU$ ($S$ edges incident to $a_4$ are not
  shown).} \label{fig:2pimp2}
\end{figure}

\subsection*{Setting up a charging scheme to lower bound the number of augmentations}
We first lay the groundwork and give a charging scheme.
\begin{note}\label{note1:2pimp2}
  For general graphs (that are possibly not triangle-free), we need to set
  $\lambdaM \ge 2$.
\end{note}
To see why, suppose $\lambdaM=1$.  Let $uv$ be a $3$-augmentable edge in
$M_0$. Then, for the edge $uv$, we might end up storing the edges $ub$ and $vb$
in $S$, and the edge $uv$ would not get augmented.  If $\lambdaM\geq 2$, and we
store at least $\lambdaM$ edges incident to $u$, then an edge incident to $v$
will not form a triangle with at least one of those and $uv$ would get
augmented.  So, for general graphs, we need to set $\lambdaM \ge 2$.

Let $|M_0| = (1/2 + \advntg) |M^*|$.  For a $3$-augmentable edge $uv \in M_0$,
let $auvb$ be the $3$-augmenting path such that $au,vb \in M^*$.  Without loss
of generality, assume that $au$ arrived before $vb$.  Then we make the following
observation.
\begin{note}
  \label{note:aunt}
  When $au$ arrived, it may not be added to $S$ for one of the following
  reasons:
  \begin{itemize}
    \item The vertex $a$ was already matched.
    \item There were $\lambdaM$ edges incident to $u$ in $S$.
    \item There were $\lambdaU$ edges incident to $a$ in $S$.
  \end{itemize}
\end{note}
We call some edges in $M_0$ \emph{good}, some \emph{partially} good, and some
\emph{bad}.  An edge is good if it got augmented.  An edge $uv \in M_0$ is bad
if it is $3$-augmentable, not good, and vertex $a$ or $b$ had $\lambdaU$ edges
incident to them in $S$ when edge $au$ or $vb$ arrived.  An edge $uv \in M_0$ is
partially good if it is $3$-augmentable, but neither good nor bad (``partially''
good because, as we will see later, we can hold some good edge $u'v'\in M_0$
responsible for $uv$ not getting augmented).  Note that all $3$-augmentable
edges get some label according to our labeling.  We require the following Lemma
to describe the charging scheme.
\begin{lem}\label{lem2:2pimp2}
  Suppose $au$ was not added to $S$ because there were already $\lambdaM$ edges
  incident to $u$ in $S$.  If, later, $uv$ did not get augmented when $vb$
  arrived, then
  \begin{itemize}
    \item $b$ was already matched via augmenting path $a''u''v''b$, or
    \item there exists $a'u\in S$ and $u'v' \in M_0$ such that $a'$ was matched
    via augmenting path $a'u'v'b'$.
  \end{itemize}
\end{lem}
\begin{proof}
  When $au$ arrived, $|N_S(u)| \ge \lambdaM$.  If $b$ was unmatched when $vb$
  arrived, then some $a' \in N_S(u) \setminus \{b\}$ must have been matched,
  otherwise we would have augmented $uv$.  Now for triangle-free graphs
  $b \notin N_S(u)$, so $|N_S(u) \setminus \{b\}| = |N_S(u)| \ge 1$, and for
  general graphs, by Observation~\ref{note1:2pimp2}, $\lambdaM \ge 2$, so
  $|N_S(u) \setminus \{b\}| \ge \lambdaM - 1 \ge 1$.
\end{proof}
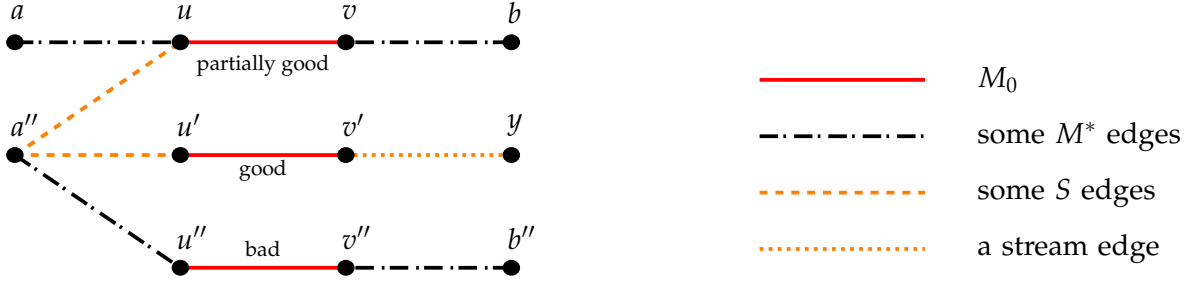
\begin{figure}
\centering
\begin{tikzpicture}[xscale=1.1,line width=0.5mm]
\draw [red] (0,0) -- (2,0);
\draw [red] (0,1.5) -- (2,1.5);
\draw [red] (0,3) -- (2,3);

\draw [dashed,orange] (-2,1.5) -- (0,3);
\draw [dashed,orange] (-2,1.5) -- (0,1.5);

\draw [dash pattern={on 7pt off 2pt on 1pt off 3pt}] (-2,3) -- (0,3);
\draw [dash pattern={on 7pt off 2pt on 1pt off 3pt}] (4,3) -- (2,3);
\draw [dash pattern={on 7pt off 2pt on 1pt off 3pt}] (-2,1.5) -- (0,0);
\draw [dash pattern={on 7pt off 2pt on 1pt off 3pt}] (4,0) -- (2,0);

\draw [dotted,orange] (4,1.5) -- (2,1.5);

\draw [red] (7,2.5) -- (9,2.5);
\draw [dash pattern={on 7pt off 2pt on 1pt off 3pt}] (7,1.75) -- (9,1.75);
\draw [dashed, orange] (7,1) -- (9,1);
\draw [dotted, orange] (7,0.25) -- (9,0.25);

\draw (12.5,2.5) node[anchor=east,text width=3cm] {$M_0$};
\draw (12.5,1.75) node[anchor=east,text width=3cm] {some $M^*$ edges};
\draw (12.5,1) node[anchor=east,text width=3cm] {some $S$ edges};
\draw (12.5,0.25) node[anchor=east,text width=3cm] {a stream edge};

\draw (1,0.4) node[anchor=east,text width=1cm] {$u''$};
\draw (3,0.4) node[anchor=east,text width=1cm] {$v''$};
\draw (5,0.4) node[anchor=east,text width=1cm] {$b''$};
\draw (1,1.9) node[anchor=east,text width=1cm] {$u'$};
\draw (-1,1.9) node[anchor=east,text width=1cm] {$a''$};
\draw (3,1.9) node[anchor=east,text width=1cm] {$v'$};
\draw (5,1.9) node[anchor=east,text width=1cm] {$y$};
\draw (1,3.4) node[anchor=east,text width=1cm] {$u$};
\draw (-1,3.4) node[anchor=east,text width=1cm] {$a$};
\draw (3,3.4) node[anchor=east,text width=1cm] {$v$};
\draw (5,3.4) node[anchor=east,text width=1cm] {$b$};

\draw (1,2.4) node[anchor=south] {\scriptsize partially good};
\draw (1,1) node[anchor=south] {\scriptsize good};
\draw (1,0) node[anchor=south] {\scriptsize bad};

\fill [color=black] (0,0) circle (3pt); 
\fill [color=black] (2,0) circle (3pt); 
\fill [color=black] (0,1.5) circle (3pt); 
\fill [color=black] (2,1.5) circle (3pt); 
\fill [color=black] (0,3) circle (3pt); 
\fill [color=black] (2,3) circle (3pt); 

\fill [color=black] (4,0) circle (3pt);
\fill [color=black] (4,1.5) circle (3pt);
\fill [color=black] (4,3) circle (3pt);
\fill [color=black] (-2,1.5) circle (3pt);
\fill [color=black] (-2,3) circle (3pt);
\end{tikzpicture}
\caption{Example showing a good edge, a bad edge, and a partially good edge.  We
  use parameters $\lambdaU = 2$ and $\lambdaM = 1$, so we are in the
  triangle-free case.  The edge $u'v'$ is not $3$-augmentable but was augmented
  using $a''u'v'y$, so $u'v'$ is a good edge.  The edge $u''v''$ is a
  $3$-augmentable edge that was not augmented and when $a''u''$ arrived,
  $\deg_S(a'') = 2$, so $u''v''$ is a bad edge.  For $uv$, we did not take $au$
  in $S$, because $\deg_S(u) = 1$, so $uv$ is a partially good edge, and we can
  charge $uv$ to $u'v'$ using~\Cref{lem2:2pimp2}.} \label{fig:2pimp2cs}
\end{figure}

\paragraph*{Charging Scheme.} As alluded to earlier, we charge a partially good
edge to some good edge.  Recall that for a $3$-augmentable edge $uv \in M_0$, we
denote by $au,vb \in M^*$ the edges that form the $3$-augmenting path with $uv$
such that $au$ arrived before $vb$.  We use Observation~\ref{note:aunt} and
consider the following cases. See~\Cref{fig:2pimp2cs}.
\begin{itemize}
  \item Suppose $au$ was not added to $S$ because $a$ was already matched.
  Then, let $u'v' \in M_0$ was augmented using $au'v'b'$.  If
  $\deg_S(a) \le \lambdaU - 1$, then we charge $uv$ to $u'v'$.  Otherwise, $uv$
  is bad.
  \item Suppose $au$ was not added to $S$ because $\deg_S(u) = \lambdaM$.  Then
  we use~\Cref{lem2:2pimp2}.  We either charge $uv$ to $u'v'$, or if
  $\deg_S(b) \le \lambdaU - 1$, then we charge $uv$ to $u''v''$.  Otherwise,
  $uv$ is bad.
  \item Suppose $au$ was not added to $S$ because $\deg_S(a)=\lambdaU$, then
  $uv$ is bad.
  \item Otherwise, $au$ was added to $S$, but $uv$ did not get augmented when $vb$
  arrived. Then:
  \begin{itemize}
    \item Either there exists $a' \in N_S(u)$ that was matched via augmenting
    path $a'u'v'b'$ (note that $a'$ may be same as $a$), then we charge $uv$ to
    $u'v'$;
    \item or $b$ was already matched via augmenting path
    $a''u''v''b$, and $vb$ was ignored; in this case, if
    $\deg_S(b) \le \lambdaU - 1$, then we charge $uv$ to $u''v''$, otherwise,
    $uv$ is bad.
  \end{itemize}
\end{itemize}
We now bound the number of bad edges in $M_0$ from above.
\begin{lem}\label{lem3:2pimp2}
  The number of bad edges is at most $\lambdaM|M_0|/\lambdaU$.
\end{lem}
\begin{proof}
  We claim that for any $uv \in M_0$, $\deg_S(u) + \deg_S(v) \le \lambdaM$,
  hence $|S| \le \lambdaM |M_0|$.  A short argument is that the
  $(\lambdaM + 1)$th edge would cause an augmentation and will not be added to
  $S$.  Let us assume the claim. By the definition of a bad edge, $\lambdaU$
  edges in $S$ are ``responsible'' for one bad edge in $M_0$.  Also, an edge
  $au'$ (or $v''b$, resp.) in $S$ can be responsible for at most one bad edge
  that can only be $uv$ if $au \notin S$ (or if $vb \notin S$, resp.;
  considering the $3$-augmenting path $auvb$).  Hence, the
  total number of bad edges is at most
  $|S|/\lambdaU \le \lambdaM|M_0|/\lambdaU$.  Now we prove the claim.

  We first prove the claim for triangle-free graphs by contradiction.  Let
  $\deg_S(u) + \deg_S(v) > \lambdaM$, and let $vy \in S$ be the
  $(\lambdaM + 1)$th edge incident to one of $u$ and $v$ that was added to $S$.
  Since $\lambdaM \ge 1$ and $\deg_S(v) \le \lambdaM$, we have
  $\deg_S(u) \ge 1$, i.e. $N_S(u) \neq \emptyset$.  Now when $vy$ arrived:
  \begin{itemize}
    \item the vertex $y$ was unmatched, otherwise $vy$ would not be added to
    $S$;
    \item no vertex $x \in N_S(u)$ was matched, otherwise $u,v \in I_B$, and
    $vy$ would not be added to $S$.
  \end{itemize}
  The above implies that when $vy$ arrived, due to some $x \in N_S(u)$ the if
  condition on~\Cref{line:c3aug} became true, and we augmented $uv$ via $xuvy$
  instead of adding $vy$ to $S$.  This is a contradiction.

  For general graphs, we argue by contradiction slightly informally for the sake
  of brevity.  By Observation~\ref{note1:2pimp2}, for general graphs,
  $\lambdaM\geq 2$.  Let $\deg_S(u) + \deg_S(v) > \lambdaM \ge 2$.  Let $vy$ be
  the second edge incident to one of $u$ and $v$ that was added to $S$; the
  first edge can be $xu$ or $vy'$.

  Suppose $xu$ was the first edge.  If $x \neq y$, then we would have augmented
  $uv$ via $xuvy$ instead of adding $vy$ to $S$---a contradiction.  If $x = y$,
  then after $vy$ was processed, $N_S(u) = N_S(v) = \{y\}$, and a third edge
  incident to one of $u$ and $v$ would not be added to $S$, because it would
  have formed a $3$-augmenting path with either $yu$ or $vy$, resulting in a
  contradiction that $\deg_S(u) + \deg_S(v) = 2$.

  Otherwise, suppose $vy'$ was the first edge; then $N_S(v) = \{y,y'\}$ after
  $vy$ was processed.  Since eventually
  $\deg_S(u) + \deg_S(v) \ge \lambdaM + 1 \ge 3$ and
  $\deg_S(u),\deg_S(v)\le \lambdaM$, we would eventually have $\deg_S(u) \ge 1$,
  so let $xu \in S$.  When $xu$ arrived, it would have formed an $3$-augmenting
  path with either $vy$ or $vy'$ (here, taking care of the fact that one of $y$
  and $y'$ can be same as $x$), resulting in a contradiction that $xu$ was not
  added to $S$.
  
  Thus, we get the claim and complete the proof.
\end{proof}
As a consequence, we get the following.
\begin{note}\label{note2:2pimp2}
  In any call to \textproc{Improve-Matching}$()$, we need to set
  $\lambdaU>\lambdaM$, i.e., $\lambdaU \ge 2$.
\end{note}
To see why, suppose $\lambdaU\le\lambdaM$.  Then by~\Cref{lem3:2pimp2},
potentially all $3$-augmentable edges in $M_0$ could become bad edges.

Recall that a $3$-augmentable edge is good, partially good, or bad; so
by~\Cref{lem:kmmlem1,lem3:2pimp2},
\begin{align}
  \text{\# good or partially good edges}
  &\geq \left(\frac{1}{2}-3\advntg\right)|M^*| -
    \frac{\lambdaM|M_0|}{\lambdaU} \nonumber\\ 
  &= \left(\frac{1}{2}-3\advntg\right)|M^*| -
    \frac{\lambdaM}{\lambdaU}\left(\frac{1}{2}+\advntg\right)|M^*| \nonumber\\ 
  &= \left( \frac{\lambdaU-\lambdaM}{2\lambdaU}
    -\left(\frac{3\lambdaU+\lambdaM}{\lambdaU}\right)\advntg\right)|M^*|\,.
    \label{eq:bndgpg}
\end{align}
In the following Lemma, we bound the number of partially good edges in $M_0$
that are charged to one good edge.
\begin{lem}\label{lem4:2pimp2}
  At most $2\lambdaU-1$ partially good edges in $M_0$ are charged to one good
  edge in $M_0$. 
\end{lem}
\begin{proof}
  Suppose $uv\in M_0$ was augmented by edges $xu$ and $vy$ such that $xu$ arrived
  before $vy$, then $xu \in S$.  Now $|N_S(x)|, |N_S(y)| \le \lambdaU$. Since
  $xu \in S$, we have $|N_S(x)\setminus \{u\}| \le \lambdaU - 1$.  Let $B :=
  (N_S(x)\setminus\{u\}) \cup N_S(y)$, then $|B| \le 2\lambdaU - 1$.  Now,
  the set of partially good edges that are charged to $uv$ is a subset of
  $M_0(B)$. Observing that $|M_0(B)| \le |B| \le 2\lambdaU -1$ finishes the proof.
\end{proof}

The following lemma characterizes the improvement given by
\textproc{Improve-Matching}$()$.
\begin{lem}\label{eq:gbound}
  Let $|M_0| = (1/2+\advntg)|M^*|$ and
  $M = \textproc{Improve-Matching}(M_0, \lambdaU, \lambdaM)$, then
  \[
    |M| \ge \left(\frac{1}{2} + \frac{\lambdaU-\lambdaM}{4\lambdaU^2} +
      \left(1-\frac{3\lambdaU+\lambdaM}{2\lambdaU^2}\right)\advntg\right)|M^*|
    \ge\left(\frac{1}{2} + \frac{\lambdaU-\lambdaM}{4\lambdaU^2}\right)|M^*|\,.
  \]
\end{lem}
\begin{proof}
  By~\eqref{eq:bndgpg} and~\Cref{lem4:2pimp2}, the total number of
  augmentations during one call to \textproc{Improve-Matching}$()$ is at least
\begin{align*}
  \frac{1}{2\lambdaU} \left( \frac{\lambdaU-\lambdaM}{2\lambdaU}
  -\left(\frac{3\lambdaU+\lambdaM}{\lambdaU}\right)\advntg\right)|M^*|
  =\left(\frac{\lambdaU-\lambdaM}{4\lambdaU^2}  
  -\left(\frac{3\lambdaU+\lambdaM}{2\lambdaU^2}\right)\advntg\right)|M^*|\,. 
\end{align*}
 
Hence, we get the following bound on the size of the output matching $M$:
\begin{align*}
  |M|&\geq |M_0| +\left(\frac{\lambdaU-\lambdaM}{4\lambdaU^2}  
       -\frac{3\lambdaU+\lambdaM}{2\lambdaU^2}\advntg\right)|M^*|\\ 
     &=\left(\frac{1}{2} + \frac{\lambdaU-\lambdaM}{4\lambdaU^2} + \left(1
       -\frac{3\lambdaU+\lambdaM}{2\lambdaU^2}\right)\advntg\right)|M^*|
     && \text{because $|M_0| = (1/2 + \advntg)|M^*|$ }\,,\\
     &\ge\left(\frac{1}{2} + \frac{\lambdaU-\lambdaM}{4\lambdaU^2}\right)|M^*|
     &&\text{since $\lambdaU \ge 2$ (see Observation~\ref{note2:2pimp2})}\,.
        \qedhere
\end{align*}
\end{proof}
Now we state and prove our main result.
\begin{theorem}
  \label{thm:2pimp2}
  \Cref{alg:2pimp2} uses two passes and has an approximation ratio of
  $1/2 + 1/16$ for tri\-angle-free graphs and an approximation ratio
  of $1/2 + 1/32$ for general graphs for maximum matching.
\end{theorem}
\begin{proof}
  By~\Cref{eq:gbound}, after the second pass, the output size
  $|M|\ge (1/2 + (\lambdaU-\lambdaM)/(4\lambdaU^2))|M^*|$; we use $\lambdaU=2$
  and $\lambdaM=1$ for triangle-free graphs and $\lambdaU=4$ and $\lambdaM=2$
  (see~Observation~\ref{note1:2pimp2}) for general graphs to get the claimed
  approximation ratios.
\end{proof}



\section{Multi Pass Algorithm}
\label{sec:multi-pass-algor}
We run the function \textproc{Improve-Matching}$()$ in~\Cref{alg:2pimp2} with
increasing values of $\lambdaU$, and the approximation ratio converges to
$1/2 + 1/6$.  We note that this multi-pass algorithm is not just a repetition of
the function \textproc{Improve-Matching}$()$.  Such a repetition will give an
asymptotically worse number of passes (see, for example, the multi-pass
algorithm due to Feigenbaum et al.~\cite{fgnbm}).  We carefully choose the
parameter $\lambdaU$ for each pass to get the required number of passes.
\begin{algorithm}[!ht]
  \caption{~~Multi-pass algorithm\label{alg:mp}: input graph $G$}
  \begin{algorithmic}[1]
    \State In the first pass, find a maximal matching $M_1$.
    \State $M\gets M_1$
    \If{$G$ is triangle-free}
    \For{$i = 2$ to $\lceil 2/(3\eps)\rceil$}
    \State $M\gets$ Improve-Matching$(M,i,1)$
    \EndFor
    \Else
    \For{$i = 2$ to $\lceil 4/(3\eps)\rceil$}
    \State $M\gets$ Improve-Matching$(M,i+1,2)$
    \EndFor
    \EndIf
    \State Return $M$.
  \end{algorithmic}
\end{algorithm}
\begin{theorem}
  \label{thm:mp}
  For any $\eps > 0$,~\Cref{alg:mp} is a semi-streaming
  $(1/2 + 1/6 - \eps)$-approximation algorithm for maximum matching that uses
  $2/(3\eps)$ passes for triangle-free graphs and $4/(3\eps)$ passes for general
  graphs.
\end{theorem}
\begin{proof}
  We prove the theorem for triangle-free case; the general case is similar.
  Let $M_i$ be the matching computed by~\Cref{alg:mp} after $i$th pass,
  and let $p:=\lceil 2/(3\eps)\rceil$, so $\eps \le 2/(3p)$.  Since $M_1$ is
  maximal, it is $(1/2)$-approximate.  Let $\advntg_1 := 0$, and for
  $i \in \{2,3,\ldots,p\}$, let
  \[
    \advntg_i := \frac{i-1}{4i^2} + \left(1 -
      \frac{3i+1}{2i^2}\right)\advntg_{i-1} \,
  \]
  (see~\Cref{eq:gbound} with $\lambdaU = i$ and $\lambdaM = 1$).  Then,
  by~\Cref{eq:gbound} and the logic of~\Cref{alg:mp}, for $i \in [p]$, the
  matching $M_i$ is $(1/2 + \advntg_i)$-approximate (by a trivial induction).
  Now we bound $\advntg_p$ by induction.  We claim
  that for $i\in [p]$, 
  \[
     \advntg_i \ge \frac{1}{6} - \frac{2}{3i}\,,
  \]
  which we prove by induction on $i$.

  Base case: For $i = 1$, we have
  $1/6 - \advntg_1 = 1/6 - 0 = 1/6 \le 2/(3\cdot 1)$.

  For inductive step, we want to show that
  \begin{align*}
    &&\frac{1}{6} - \advntg_i =  \frac{1}{6} - \frac{i-1}{4i^2} - \left(1 -
    \frac{3i+1}{2i^2}\right)\advntg_{i-1} &\le \frac{2}{3i}\,,\\
    \intertext{\text{which is implied by the following (using inductive
    hypothesis)}} 
    &&\frac{1}{6} - \frac{i-1}{4i^2} + \left(1 -
    \frac{3i+1}{2i^2}\right)\left(\frac{2}{3(i-1)}-\frac{1}{6}\right)
                                          &\le \frac{2}{3i}\,,\\
    \text{which is implied by} &&\frac{1}{6} - \frac{i-1}{4i^2} +
    \left(\frac{2i^2 -3i -1}{2i^2}\right)\left(\frac{4-i+1}{6(i-1)}\right) 
                                          &\le \frac{2}{3i}\,,\\
    \intertext{multiplying both sides by $12i^2(i-1)$, we then need to show
    that,} 
    &&2i^2(i-1) - 3(i-1)^2 + (2i^2-3i-1)(-i+5)&\le 8i(i-1)\\
    \text{which is implied by}&&
    2i^3-2i^2 - 3(i^2-2i+1) + (-2i^3+10i^2+3i^2-15i+i-5)&\le 8i^2-8i\\
    \text{which is implied by}&&
    2i^3-5i^2 +6i-3 + (-2i^3+13i^2-14i-5)&\le 8i^2-8i\\
    \text{which is implied by}&&
    8i^2-8i -8 &\le 8i^2-8i
  \end{align*}
  which is true, so we get the claim.  Therefore
  $\advntg_p\ge 1/6 - 2/(3p) \ge 1/6 - \eps$, and by our earlier observation,
  $M_p$ is $(1/2 + \advntg_p)$-approximate, and this finishes the proof for
  triangle-free case.  The proof for general case is very similar.  We define
  $p := \lceil 4/(3\eps)\rceil$ and $\advntg_1 := 0$, and for
  $i \in \{2,3,\ldots,p\}$, we define
  \[
    \advntg_i := \frac{i-1}{4(i+1)^2} + \left(1 -
      \frac{3(i+1)+2}{2(i+1)^2}\right)\advntg_{i-1} \,,
  \]
  i.e., we use $\lambdaU = i + 1$ and $\lambdaM = 2$.  The corresponding claim
  then is that for $i\in [p]$,
  \[
    \advntg_i \ge \frac{1}{6} - \frac{4}{3i}\,,
  \]
  which can be verified by induction on $i$.
\end{proof}


\subparagraph*{Acknowledgements.}  We thank Sundar Vishwanathan and Ashish
Chiplunkar for helpful discussions.  The first author would like to thank his
advisor Amit Chakrabarti and Andrew McGregor for helpful discussions.


\bibliographystyle{plain}
\bibliography{references_mcm}
\appendix
\section{Three Pass Algorithm for Triangle Free Graphs}
\label{sec:3ptf}
For completeness, we present our three-pass algorithm for triangle-free graphs.
\algrenewcommand\algorithmicforall{\textbf{foreach}}
\begin{algorithm}[!ht]
  \caption{~Three-pass algorithm for triangle-free graphs \label{alg:3ptf}}
  \begin{algorithmic}[1]
    \State In the first pass, find a maximal matching $M_0$.
    \State In the second pass, find a maximal matching $M_1$ in $F_1 := \{uv : u
    \in V\setminus V(M_0), v\in V(M_0)\}$.
    \State After the second pass:
    \begin{itemize}
     \item $M'_1\gets$ arbitrary largest subset of $M_1$ such that there is no
      $3$-augmenting path in $M'_1 \cup M_0$ with respect to $M_0$
    \item $V_2\gets \{x \in V(M_0) : \exists v,w \text{ such that } vw \in M'_1
      \text{ and } wx \in M_0\} $
    \item For $x \in V_2$, denote by $P(x)$ the vertex $v$ such that there
      exists $w$ with $vw \in M'_1$ and $wx\in M_0$. See $x$ and $P(x)$
      in~\Cref{fig:2ptf}.
    \end{itemize}
    \State In the third pass: $F_2 := \{xy : x \in V_2, y \in V\setminus
    V(M_0)\}$ 
    \State $M_2\gets\emptyset$
    \For {edge $xy \in F_2$}
    \If{$x$, and $y$ are unmarked}
    \State $M_2\gets M_2\cup\{xy\}$; since the graph is triangle free, $y \neq
    P(x)$, and we can augment $M_0$ using $xy$.
    \State Mark $P(x)$, $x$, $y$, and $P^{-1}(y)$ (if exists).
    \EndIf
    \EndFor
    \State Let $M$ be largest of $M_3$ and $M'_3$ which are computed below.
    \begin{itemize}
    \item Augment $M_0$ using edges in $M_1$ to get $M_3$.
    \item Augment $M_0$ using edges in $M'_1$ and $M_2$ to get $M'_3$.
    \end{itemize}
    \State  Output $M$.
  \end{algorithmic}
\end{algorithm}
\begin{theorem}
  \label{thm:3ptf}
  \Cref{alg:3ptf} is a three-pass, semi-streaming, $(1/2 +1/10)$-approximation
  algorithm for maximum matching in triangle-free graphs, and the analysis is
  tight.
\end{theorem}
\begin{proof}
  Let $|M_0| = (1/2 + \advntg) |M^*|$.  The number of edges in $M^*$
  incident on $V(M^*)\setminus V(M_0)$ is
  \begin{equation}
    \label{eq:mubound}
    |V(M^*)\setminus V(M_0)|
    \ge |V(M^*)|- |V(M_0)|
    = 2|M^*| - 2|M_0|
    = (1-2\advntg)|M^*|\,;    
  \end{equation}
  and these edges also belong to $F_1$. Since $M_1$ is a maximal matching in
  $F_1$,
  \begin{equation}
    \label{eq:3}
    |M_1| \ge (1-2\advntg)|M^*|/2 = (1/2-\advntg)|M^*|\,.
  \end{equation}
  Let $c$ be the number of $3$-augmenting paths in $M_1 \cup M_0$, so
  $|M'_1| = |M_1| - c$ by the definition of $M'_1$.  By~\Cref{lem:kmmlem1},
  there are at most $4\advntg|M^*|$ non-$3$-augmentable edges in $M_0$.  So at
  least $|M_1| - c - 4\advntg|M^*|$ edges of $M'_1$ are incident on
  $3$-augmentable edges of $M_0$.  Therefore there is a matching of size at
  least $|M_1| - c - 4\advntg|M^*|$ in $F_2$; consider one such matching $M_F$.
  We claim that $|M_2| \ge |M_F|/4$.  See~\Cref{fig:2ptf}.  Let $xy \in M_2$; we
  note that $xy$ disallows at most four edges in $M_F$ from being added to $M_2$
  due to the (at most) four marks that it adds, because a marked vertex can
  disallow at most one edge in $M_F$ (due to it being a matching), which shows
  the claim.  Hence: 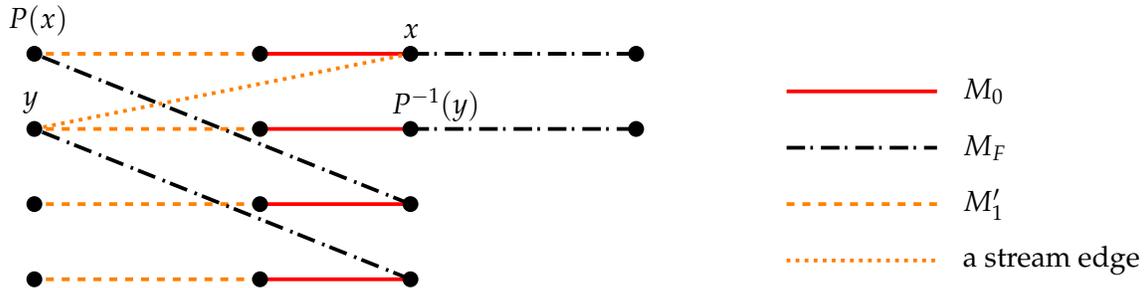
\begin{figure}
\centering
\begin{tikzpicture}[line width=0.5mm]

\draw [red] (0,0) -- (2,0);
\draw [red] (0,1) -- (2,1);
\draw [red] (0,2) -- (2,2);
\draw [red] (0,3) -- (2,3);

\draw [dash pattern={on 7pt off 2pt on 1pt off 3pt}] (2,2) -- (5,2);
\draw [dash pattern={on 7pt off 2pt on 1pt off 3pt}] (2,3) -- (5,3);

\draw [dashed,orange] (0,0) -- (-3,0);
\draw [dashed,orange] (0,1) -- (-3,1);
\draw [dashed,orange] (0,2) -- (-3,2);
\draw [dashed,orange] (0,3) -- (-3,3);

\draw [dash pattern={on 7pt off 2pt on 1pt off 3pt}] (2,1) -- (-3,3);
\draw [dash pattern={on 7pt off 2pt on 1pt off 3pt}] (2,0) -- (-3,2);

\draw [dotted,orange] (-3,2) -- (2,3);

\draw [red] (7,2.5) -- (9,2.5);
\draw [dash pattern={on 7pt off 2pt on 1pt off 3pt}] (7,1.75) -- (9,1.75);
\draw [dashed, orange] (7,1) -- (9,1);
\draw [dotted,orange] (7,0.25) -- (9,0.25);

\draw (10.5,2.5) node[anchor=east,text width=1cm] {$M_0$};
\draw (10.5,1.75) node[anchor=east,text width=1cm] {$M_F$};
\draw (10.5,1) node[anchor=east,text width=1cm] {$M_1'$};
\draw (11.85,0.25) node[anchor=east] {a stream edge};

\draw (1.75,3.55) node[anchor=north west] {$x$};
\draw (-3.5,3.8) node[anchor=north west] {$P(x)$};
\draw (-3.3,2.65) node[anchor=north west] {$y$};
\draw (1.6,2.70) node[anchor=north west] {$P^{-1}(y)$};

\fill [color=black] (0,0) circle (3pt);
\fill [color=black] (2,0) circle (3pt);
\fill [color=black] (0,1) circle (3pt);
\fill [color=black] (2,1) circle (3pt);
\fill [color=black] (0,2) circle (3pt);
\fill [color=black] (2,2) circle (3pt);
\fill [color=black] (0,3) circle (3pt);
\fill [color=black] (2,3) circle (3pt);

\fill [color=black] (-3,0) circle (3pt);

\fill [color=black] (-3,1) circle (3pt);

\fill [color=black] (-3,2) circle (3pt);
\fill [color=black] (5,2) circle (3pt);
\fill [color=black] (-3,3) circle (3pt);
\fill [color=black] (5,3) circle (3pt);

\end{tikzpicture}
\caption{An edge $xy \in M_2$ disallows at most four edges in $M_F$ from being
  added to $M_2$.}
  \label{fig:2ptf}
\end{figure}

  \begin{align*}
    |M_2| &\ge \frac{|M_F|}{4}\\
          &\ge \frac{|M_1| - c - 4\advntg|M^*|}{4}\\
          &\ge \frac{1}{4}\left(\left(\frac{1}{2}-\advntg\right)|M^*| - c -
            4\advntg|M^*|\right)
          &&\text{by~\eqref{eq:3}}\,,\\
          &= \frac{1}{4}\left(\left(\frac{1}{2}- 5\advntg\right)|M^*| - c \right)\,.
  \end{align*}
  Now, each edge in $M_2$ gives one augmentation after the second pass.  To see
  this, we observe that for any $x\in V_2$, at any point in the algorithm, $x$
  and $P(x)$ are either both marked or both unmarked.  So when an edge
  $xy\in M_2$ arrives, $x$ and $y$ are unmarked, and $P(x)$ and $P^{-1}(y)$ (if
  it exists) are also unmarked, otherwise one of $x$ and $y$ would have been
  marked and $xy$ would not have been added to $M_2$.  Since both $P(x)$ and
  $P^{-1}(y)$ were unmarked, we can use the augmenting path
  $\{M'_1(\{P(x)\}),M_0(\{x\}),xy\}$.  Hence we get at least
  \[
    \max\left\{c,  \frac{1}{4}\left(\left(\frac{1}{2}- 5\advntg\right)|M^*| - c
      \right) \right\} 
  \]
  augmentations after the third pass.  This is minimized by setting
  \begin{align*}
    c &= \frac{1}{4}\left(\left(\frac{1}{2}- 5\advntg\right)|M^*| - c \right)\\
    &=\frac{1}{5}\left(\left(\frac{1}{2}- 5\advntg\right)|M^*| \right)\\
    &=\left(\frac{1}{10}- \advntg\right)|M^*| \,.
  \end{align*}
  So we get the following bound:
  \[
      |M| \ge |M_0| +\left(\frac{1}{10} - \advntg\right)|M^*|
        \ge \left(\frac{1}{2} + \advntg\right)|M^*| +\left(\frac{1}{10} -
          \advntg\right)|M^*|
        = \left(\frac{1}{2} + \frac{1}{10} \right)|M^*|\,.  \qedhere
  \]
  The tight example is shown in~\Cref{fig:3ptfte}.
\end{proof}
\begin{figure}
\centering
\begin{tikzpicture}[line width=0.5mm]

\draw [red] (0,0) -- (2,0);
\draw [red] (0,1) -- (2,1);
\draw [red] (0,2) -- (2,2);

\draw [dash pattern={on 7pt off 2pt on 1pt off 3pt}] (2,0) -- (4,0);
\draw [dash pattern={on 7pt off 2pt on 1pt off 3pt}] (2,1) -- (4,1);
\draw [dash pattern={on 7pt off 2pt on 1pt off 3pt}] (-2,2) -- (0,2);
\draw [dash pattern={on 7pt off 2pt on 1pt off 3pt}] (-2,0) -- (0,0);
\draw [dash pattern={on 7pt off 2pt on 1pt off 3pt}] (0,1) -- (2,2);

\draw [dashed,orange] (0,2) -- (-2,0);
\draw [dashed,orange] (2,1) -- (4,0);

\draw [red] (6,1.85) -- (8,1.85);
\draw [dash pattern={on 7pt off 2pt on 1pt off 3pt}] (6,1.1) -- (8,1.1);
\draw [dashed, orange] (6,0.35) -- (8,0.35);

\draw (9.5,1.85) node[anchor=east,text width=1cm] {$M_0$};
\draw (9.5,1.1) node[anchor=east,text width=1cm] {$M^*$};
\draw (9.5,0.35) node[anchor=east,text width=1cm] {$M_1$};

\draw (0,3) node[anchor=north] {\scriptsize $A(M_0)$};
\draw (-2,3) node[anchor=north] {\scriptsize $B\setminus B(M_0)$};
\draw (2,3) node[anchor=north] {\scriptsize $B(M_0)$};
\draw (4,3) node[anchor=north] {\scriptsize $A\setminus A(M_0)$};

\fill [color=black] (0,0) circle (3pt);
\fill [color=black] (2,0) circle (3pt);
\fill [color=black] (0,1) circle (3pt);
\fill [color=black] (2,1) circle (3pt);
\fill [color=black] (0,2) circle (3pt);
\fill [color=black] (2,2) circle (3pt);

\fill [color=black] (-2,0) circle (3pt);
\fill [color=black] (4,1) circle (3pt);
\fill [color=black] (-2,2) circle (3pt);
\fill [color=black] (4,0) circle (3pt);

\end{tikzpicture}
\caption{Tight example for~\Cref{alg:3ptf}: $M_1$ has only two edges that land
  on bad vertices and cannot be augmented in the third pass. So
  $|M| = |M_0| = 3$ and $|M^*| = 5$.}
  \label{fig:3ptfte}
\end{figure}
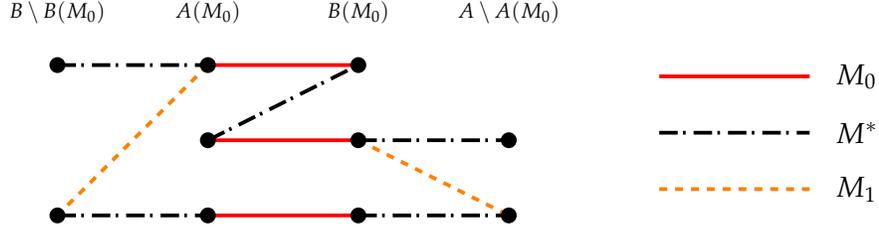


\section{Three Pass Algorithm for General Graphs}
\label{sec:3pgeneral}
We find a maximal matching $M_1$ in the first pass.  Then we use
\textproc{Improve-Matching}$()$ function from~\Cref{alg:2pimp2}, i.e.,
\begin{itemize}
  \item in the second pass, $M_2\gets\textproc{Improve-matching}(M_1,4,2)$, and
  \item in the third pass, $M_3\gets\textproc{Improve-matching}(M_2,5,2)$.
\end{itemize}
We observe that $M_1$ is $(1/2)$-approximate. Then by double application
of~\Cref{eq:gbound}, we get that $M_3$ is $(1/2 + 81/1600)\approx (1/2 + 1/19.753)$-approximate.

\section{Three Pass Algorithm for Bipartite Graphs: Suboptimal Analysis}
We now give an analysis of~\Cref{alg:3pkmm} that shows approximation ratio of
only $1/2 + 1/18$ that is based on Konrad et al.'s~\cite{kmm} analysis for their
two-pass algorithm for bipartite graphs.  Afterward, we demonstrate that by not
considering the distribution of lengths of augmenting paths, we may prove an
approximation ratio of at most $1/2 + 1/14$.  The better and tight analysis
appears in~\Cref{sec:3pkmm}.
\begin{theorem}
  \label{thm:3pkmm}
  \Cref{alg:3pkmm} is a three-pass, semi-streaming,
  $(1/2 +1/18)$-approximation algorithm for maximum matching in bipartite
  graphs.
\end{theorem}
\begin{proof}
  As usual, let $|M_0| = (1/2 + \advntg) |M^*|$.  Since $M_0$ is a maximal
  matching, there are $|B(M^*) \setminus B(M_0)|$ edges of $M^*$ that are also
  in $F_2$.  We have
  \[
    |B(M^*) \setminus B(M_0)|\ge |B(M^*)|-|B(M_0)|= |M^*| - |M_0|\,,
  \]
  and since $M_A$ is maximal, we then get the following:
  \begin{equation}
    \label{eq:mabound}
    |M_A| \ge \frac{1}{2}|B(M^*) \setminus B(M_0)|
    \ge \frac{1}{2}(|M^*| - |M_0|)
    = \frac{1}{2}\left(1  - \left(\frac{1}{2} + \advntg\right)\right)|M^*|    
    = \frac{1}{2}\left(\frac{1}{2}  - \advntg\right)|M^*|.    
  \end{equation}
  By~\Cref{lem:kmmlem1}, there are at most $4\advntg|M^*|$ non-$3$-augmentable
  edges in $M_0$.  Which means that at least $|M_A| - 4\advntg|M^*|$ edges of
  $M_A$ are incident on $3$-augmentable edges of $M_0$; therefore there is a
  matching of size at least $|M_A| - 4\advntg|M^*|$ in $F_3$.  Since we output a
  maximal matching in $F_3$, we get at least $(1/2)(|M_A| - 4\advntg|M^*|)$
  augmentations after the third pass.  So we get the following bound:
  \begin{align*}
    |M| &\ge |M_0| + \frac{1}{2}(|M_A| - 4\advntg|M^*|)\\
        &\ge |M_0| + \frac{1}{2}\left(\frac{1}{2}\left(\frac{1}{2} - \advntg\right)
          - 4\advntg\right)|M^*| &&\text{ by~\eqref{eq:mabound}}\,,\\
        &= |M_0| + \left(\frac{1}{8} -\frac{9}{4}\advntg\right)|M^*|\\
        &=  \left(\frac{1}{2}  + \advntg\right)|M^*| +
          \left(\frac{1}{8} -\frac{9}{4}\advntg\right)|M^*|
        &&\text{ because $|M_0| = (1/2 + \advntg) |M^*|$}\,,\\
        &= \left(\frac{1}{2} + \frac{1}{8} - \frac{5}{4}\advntg\right)|M^*|\,.    
  \end{align*}
  We also have $|M| \ge |M_0| = (1/2 + \advntg)|M^*|$.  As $\advntg$ increases, the
  former bound deteriorates and the latter improves, so the worst case $\advntg$ is
  when these two bounds are equal, which happens at $\advntg = 1/18$, and the
  approximation ratio we get is $1/2 + 1/18$.
\end{proof}

\subsection{Improved Analysis Without Considering Longer Augmenting Paths}
\label{sec:three-pass-algorithm}
We can analyze~\Cref{alg:3pkmm} better if we bound $|M_A|$ more carefully.  The
claim is that at least $(1/2 - 7\advntg)|M^*|/2$ edges of $M_A$ are incident on
$3$-augmentable edges of $M_0$.  Let $A_G \subseteq A(M_0)$ be the set of
vertices in $A$ that are endpoints of $3$-augmentable edges of $M_0$; also, let
$A_N = A(M_0)\setminus A_G$.  So there is a matching of size at least $|A_G|$ in
$F_2$ that covers $A_G$.  Any maximal matching in $F_2$ has at least
$(|A_G| - |A_N|)/2$ edges that are incident on $A_G$.  To see the
claim, we use the facts $|A_G| \ge (1/2 - 3\advntg)|M^*|$ and
$|A_N| \le 4\advntg|M^*|$.  So there is a matching of size at least
$(1/2 - 7\advntg)|M^*|/2$ in $F_3$.  We output a maximal matching in $F_3$; hence
we get at least $(1/2 - 7\advntg)|M^*|/4$ augmentations after the third pass.  So
we get the following bound:
\begin{align*}
  |M| &\ge |M_0| + \frac{1}{4}\left(\frac{1}{2} - 7\advntg\right)|M^*|\\
      &=  \left(\frac{1}{2}  + \advntg\right)|M^*| +
        \frac{1}{4}\left(\frac{1}{2} - 7\advntg\right)|M^*|\\
      &= \left(\frac{1}{2} + \frac{1}{8} - \frac{3}{4}\advntg\right)|M^*|\,.     
\end{align*}
where the second inequality is by~\eqref{eq:mabound}.  We also have
$|M| \ge |M_0| = (1/2 + \advntg)|M^*|$, so the worst case $\advntg$ is when these two
bounds are equal, which happens at $\advntg = 1/14$ and the approximation ratio we
get is $1/2 + 1/14$, and we get the following theorem.  
\begin{theorem}
  \label{thm:3pkmmimp}
  \Cref{alg:3pkmm} is a three-pass, semi-streaming,
  $(1/2 +1/14)$-approximation algorithm for maximum matching in bipartite
  graphs.  
\end{theorem}


\section{A Note on the Analysis by Esfandiari et al.~\cite{Esf2017}}
\label{sec:issue-gener-lemma}
We demonstrate with an example that the analysis of the algorithm by Esfandiari
et al.~\cite{Esf2017} given for bipartite graphs cannot be extended for
triangle-free graphs to get the same approximation ratio.
See~\Cref{fig:ehmcetf}.  Lemma~6 in their paper, as they correctly claim, holds
only for bipartite graphs and not for triangle-free graphs.  Our algorithm
in~\Cref{sec:simple-two-pass} is essentially the same algorithm except for the
post-processing step; we augment the maximal matching computed in the first pass
greedily, whereas they use an offline maximum matching algorithm.  We have
highlighted some other comparison points in~\Cref{par:note}.

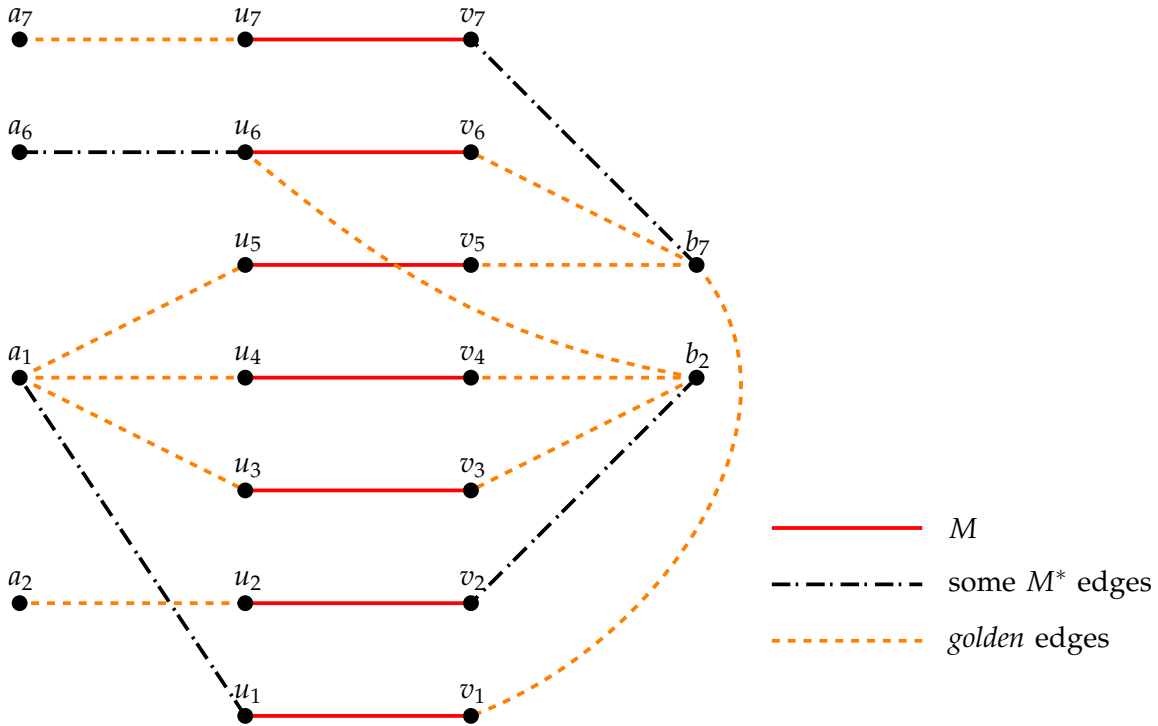
\begin{figure}
\centering
\begin{tikzpicture}[line width=0.5mm]
\draw [red] (0,0) -- (3,0);
\draw [red] (0,1.5) -- (3,1.5);
\draw [red] (0,3) -- (3,3);
\draw [red] (0,4.5) -- (3,4.5);
\draw [red] (0,6) -- (3,6);
\draw [red] (0,7.5) -- (3,7.5);
\draw [red] (0,9) -- (3,9);

\draw [dash pattern={on 7pt off 2pt on 1pt off 3pt}]  (-3,7.5) -- (0,7.5);
\draw [dash pattern={on 7pt off 2pt on 1pt off 3pt}]  (0,0) -- (-3,4.5);
\draw [dash pattern={on 7pt off 2pt on 1pt off 3pt}]  (6,4.5) -- (3,1.5);
\draw [dash pattern={on 7pt off 2pt on 1pt off 3pt}]  (6,6) -- (3,9);

\draw [dashed,orange]  (6,6) -- (3,7.5);
\draw [dashed,orange]  (6,6) -- (3,6);
\draw [dashed,orange]  (6,6) to[out=-50,in=20] (3,0);
\draw [dashed,orange]  (6,4.5) to[out=170,in=-40] (0,7.5);
\draw [dashed,orange]  (6,4.5) -- (3,4.5);
\draw [dashed,orange]  (6,4.5) -- (3,3);
\draw [dashed,orange] (-3,9) -- (0,9);
\draw [dashed,orange] (-3,4.5) -- (0,6);
\draw [dashed,orange] (-3,4.5) -- (0,4.5);
\draw [dashed,orange] (-3,4.5) -- (0,3);
\draw [dashed,orange] (0,1.5) -- (-3,1.5);

\draw [red] (7,2.5) -- (9,2.5);
\draw [dash pattern={on 7pt off 2pt on 1pt off 3pt}] (7,1.75) -- (9,1.75);
\draw [dashed, orange] (7,1) -- (9,1);

\draw (12.5,2.5) node[anchor=east,text width=3cm] {$M$};
\draw (12.5,1.75) node[anchor=east,text width=3cm] {some $M^*$ edges};
\draw (12.5,1) node[anchor=east,text width=3cm] {\emph{golden} edges};

\draw (1,0.3) node[anchor=east,text width=1cm] {$u_1$};
\draw (4,0.3) node[anchor=east,text width=1cm] {$v_1$};
\draw (1,1.8) node[anchor=east,text width=1cm] {$u_2$};
\draw (4,1.8) node[anchor=east,text width=1cm] {$v_2$};
\draw (1,3.3) node[anchor=east,text width=1cm] {$u_3$};
\draw (4,3.3) node[anchor=east,text width=1cm] {$v_3$};
\draw (1,4.8) node[anchor=east,text width=1cm] {$u_4$};
\draw (4,4.8) node[anchor=east,,text width=1cm] {$v_4$};
\draw (1,6.3) node[anchor=east,text width=1cm] {$u_5$};
\draw (4,6.3) node[anchor=east,text width=1cm] {$v_5$};
\draw (1,7.8) node[anchor=east,text width=1cm] {$u_6$};
\draw (4,7.8) node[anchor=east,text width=1cm] {$v_6$};
\draw (1,9.3) node[anchor=east,text width=1cm] {$u_7$};
\draw (4,9.3) node[anchor=east,text width=1cm] {$v_7$};

\draw (-2,9.3) node[anchor=east,text width=1cm] {$a_7$};
\draw (-2,7.8) node[anchor=east,text width=1cm] {$a_6$};
\draw (-2,4.8) node[anchor=east,text width=1cm] {$a_1$};
\draw (-2,1.8) node[anchor=east,text width=1cm] {$a_2$};

\draw (7,4.8) node[anchor=east,,text width=1cm] {$b_2$};
\draw (7,6.3) node[anchor=east,text width=1cm] {$b_7$};

\fill [color=black] (0,0) circle (3pt); 
\fill [color=black] (3,0) circle (3pt); 
\fill [color=black] (0,1.5) circle (3pt); 
\fill [color=black] (3,1.5) circle (3pt); 
\fill [color=black] (0,3) circle (3pt); 
\fill [color=black] (3,3) circle (3pt); 
\fill [color=black] (0,4.5) circle (3pt); 
\fill [color=black] (3,4.5) circle (3pt); 
\fill [color=black] (0,6) circle (3pt); 
\fill [color=black] (3,6) circle (3pt); 
\fill [color=black] (0,7.5) circle (3pt); 
\fill [color=black] (3,7.5) circle (3pt); 
\fill [color=black] (0,9) circle (3pt); 
\fill [color=black] (3,9) circle (3pt); 

\fill [color=black] (-3,4.5) circle (3pt); 
\fill [color=black] (-3,7.5) circle (3pt); 
\fill [color=black] (-3,9) circle (3pt); 
\fill [color=black] (6,6) circle (3pt); 
\fill [color=black] (6,4.5) circle (3pt); 
\fill [color=black] (-3,1.5) circle (3pt); 

\end{tikzpicture}
\caption{Example demonstrating that Lemma~6 in Esfandiari et al.~\cite{Esf2017}
  does not hold when the input graph is not bipartite but is triangle-free.  We
  use $k = 3$.  For an $M$ edge $u_iv_i$, there are two $M^*$ edges incident on
  it, which are $a_iu_i$ and $v_ib_i$, and some of the $M^*$ edges are not
  shown, but all golden edges are shown, which we call support edges or denote
  by $S$ in our terminology.  It can be seen from this example that their
  algorithm is not a $(1/2+1/12)$-approximation algorithm for triangle free
  graphs, because out of the seven $3$-augmentable edges in $M$, only one will
  get augmented, thereby giving a worse approximation
  ratio.} \label{fig:ehmcetf}
\end{figure}



\end{document}